\documentclass[11pt]{article}

\usepackage[T1]{fontenc}
\usepackage[letterpaper,left=1.05in,right=1.05in,top=1.20in,bottom=1.05in]{geometry}
\usepackage{newtxtext}
\usepackage{amsmath,bm,mathtools}
\usepackage{amsthm}
\usepackage{newtxmath}
\usepackage{array,booktabs,tabularx}
\usepackage{graphicx}
\usepackage{microtype}
\usepackage{siunitx}
\usepackage{xcolor}
\usepackage{textcase}
\usepackage[compact]{titlesec}
\usepackage[font=small,labelfont={bf,small},labelsep=period,justification=justified]{caption}
\usepackage[numbers,sort&compress]{natbib}
\usepackage{url}
\usepackage{enumitem}
\usepackage{float}
\usepackage{rotating}
\usepackage[pdfusetitle]{hyperref}
\usepackage{cleveref}

\hypersetup{
  colorlinks=true,
  allcolors={blue!55!black},
  urlcolor={blue!55!black},
  breaklinks=true
}

\newcommand{\doi}[1]{\href{https://doi.org/#1}{\nolinkurl{doi:#1}}}

\renewcommand{\theHsection}{\arabic{section}}
\renewcommand{\theHsubsection}{\arabic{section}.\arabic{subsection}}
\renewcommand{\theHsubsubsection}{\arabic{section}.\arabic{subsection}.\arabic{subsubsection}}

\renewcommand{\thesection}{\Roman{section}}
\renewcommand{\thesubsection}{\Alph{subsection}}

\titleformat{\section}[block]
  {\centering\normalfont\normalsize\bfseries}
  {\thesection.}{0.6em}{\MakeTextUppercase}
\titleformat{name=\section,numberless}[block]
  {\centering\normalfont\normalsize\bfseries}
  {}{0em}{\MakeTextUppercase}
\titleformat{\subsection}[hang]
  {\normalfont\normalsize\itshape}
  {\thesubsection.}{0.5em}{}
\titlespacing*{\section}{0pt}{1.35em}{0.65em}
\titlespacing*{\subsection}{0pt}{0.95em}{0.4em}

\makeatletter
\renewcommand{\@maketitle}{%
  \newpage\null
  \vskip 1.4em%
  \begin{center}%
    {\LARGE\bfseries\boldmath \@title \par}%
    \vskip 0.95em%
    {\large \@author \par}%
    \vskip 0.35em%
    {\normalsize Independent researcher\par}%
    \vskip 0.25em%
    {\small ORCID: \href{https://orcid.org/0009-0003-0046-821X}{0009-0003-0046-821X}\par}%
  \end{center}%
  \par\vskip 1.1em}
\renewenvironment{abstract}{%
  \begin{center}
    {\small\bfseries ABSTRACT}\\[0.55em]
    \begin{minipage}{0.92\textwidth}
    \small
    \ignorespaces
}{%
    \end{minipage}
  \end{center}
  \vspace{1.05em}}
\makeatother

\newcounter{algorithm}

\crefname{algorithm}{Algorithm}{Algorithms}
\Crefname{algorithm}{Algorithm}{Algorithms}
\crefname{figure}{Fig.}{Figs.}
\Crefname{figure}{Fig.}{Figs.}
\crefname{table}{Table}{Tables}
\Crefname{table}{Table}{Tables}
\crefname{section}{Sec.}{Secs.}
\Crefname{section}{Sec.}{Secs.}
\crefname{appendix}{Appendix}{Appendixes}
\Crefname{appendix}{Appendix}{Appendixes}

\theoremstyle{plain}
\newtheorem{proposition}{Proposition}

\setlist{nosep}
\newcommand{\vect}[1]{\bm{#1}}
\newcommand{\mS}{\mathsf{S}}
\newcommand{\cE}{\mathcal{E}}
\newcommand{\cC}{\mathcal{C}}
\newcommand{\cQ}{\mathcal{Q}}
\newcommand{\cR}{\mathcal{R}}
\newcommand{\cD}{\mathcal{D}}
\newcommand{\cS}{\mathcal{S}}
\newcommand{\cX}{\mathcal{X}}
\newcommand{\cH}{\mathcal{H}}
\newcommand{\Lat}{\mathcal{L}}
\newcommand{\ket}[1]{|#1\rangle}

\newcommand{\braket}[2]{\langle#1|#2\rangle}
\DeclareMathOperator{\Tr}{Tr}
\newcommand{\ketbra}[2]{|#1\rangle\!\langle#2|}

\title{Digital quantum lattice Boltzmann evolution by reversible compute and open-system reset}
\author{Muhammad Idrees Khan}
\date{}

\begin{document}
\maketitle

\begin{abstract}
Turbulent fluid simulation is computationally demanding because nonlinear interactions couple a wide range of scales.
Quantum computing offers another computational model, but fluid evolution is nonlinear and dissipative whereas closed gate-based dynamics is linear and reversible.
The lattice Boltzmann method (LBM) is an attractive discrete target because streaming is a local permutation, yet conventional collision reconstructs a nonlinear equilibrium and applies dissipative relaxation.
We present a digital quantum lattice Boltzmann (QLBM) algorithm that keeps that nonlinear timestep in computational-basis registers.
Each step writes the updated populations into a clean destination bank, uncomputes the workspace, and resets the obsolete source.
The resulting channel is completely positive and trace preserving (CPTP), composes without intermediate measurement, and applies to any lattice Boltzmann stencil that can be evaluated reversibly into a clean bank.
Forced D3Q19 homogeneous isotropic turbulence, in direct numerical simulation (DNS) and Smagorinsky large-eddy simulation (LES) at Reynolds number 15000 over one million lattice steps, shows statistical agreement in mass, energy budgets, spectra, and intermittency.
The three-dimensional calculations use a digital-register emulator of the circuit.
A separate D2Q9 implementation compiles the same contract to reversible gates and agrees bit-exactly with its integer twin.
Two-bank storage and serial lattice depth limit near-term feasibility.
\end{abstract}

\section{Introduction}
Fluid simulation is computationally demanding in turbulence, where
nonlinear interactions couple a broad range of scales
\cite{pope2000turbulent}.
Direct numerical simulation (DNS) resolves those scales explicitly, at a
cost that grows rapidly with Reynolds number.
Large-eddy simulation (LES) reduces that cost by modelling unresolved
motion with a state-dependent subgrid closure.

Quantum computing offers a different computational model and has motivated
algorithms for fluid dynamics.
Closed gate-based evolution is linear and reversible, whereas fluid
evolution is nonlinear and dissipative.
A quantum fluid algorithm must therefore represent irreversible relaxation
over repeated timesteps.

The lattice Boltzmann method (LBM) replaces a direct Navier--Stokes
discretization by the evolution of discrete particle populations through
site-local collision and nearest-neighbour streaming
\cite{succi2018lbm,kruger2017lbm}.
Streaming is a reversible permutation, and macroscopic density and momentum
are local moments of the populations.
The obstacle for a quantum implementation is collision.
Conventional hydrodynamic LBM reconstructs a nonlinear equilibrium and
applies dissipative relaxation.
Practical solvers add body forcing, a choice of relaxation rates, and, for
turbulent LES, locally state-dependent dissipation.

Quantum lattice Boltzmann methods (QLBMs) treat that mismatch by
simplifying the collision model, lifting it into a larger linear system, or
changing the representation of collision
\cite{wawrzyniak2025qlbm,wawrzyniak2025dynamic,nagel2026multistep,sanavio2024carleman,sanavio2024routes,itani2024qalb,bastida2026carleman,wang2025nonlinear}.
\Cref{sec:sota} compares encodings and measurement contracts in detail.

Here a conventional nonlinear LBM timestep is written as a deterministic
quantum channel and iterated without decoding and rebuilding the fluid
after every step.
Any deterministic digital map $\Phi$ that can be reversibly evaluated into a
clean destination bank defines a two-bank compute--reset timestep.
The circuit computes $\Phi$, uncomputes the workspace, and unconditionally
resets the obsolete source bank.
The resulting channel is completely positive and trace preserving (CPTP),
reproduces $\Phi$ on computational-basis inputs, and composes without
intermediate readout (\cref{prop:cptp-reset}).
Bhatnagar--Gross--Krook (BGK), two-relaxation-time (TRT), and
multiple-relaxation-time (MRT) collision all use that contract, with one
rate, two rates, or a full diagonal of rates.
Bank reset dephases in the computational basis, so the algorithm is a
digital computational-basis method rather than an amplitude-encoded or
coherent QLBM.

The test problem is forced homogeneous isotropic turbulence (HIT), run as
DNS and as LES.
Forced D3Q19 HIT is the three-dimensional application, on a $320^3$ DNS
grid and on a $64^3$ Smagorinsky LES grid with a locally state-dependent
closure evaluated from the evolving register
\cite{yu2006mrtles,smagorinsky1963general,premnath2009forcing,chai2012mrtstrain}.
A separate D2Q9 implementation compiles the same compute--reset timestep
to reversible gates and executes those circuits exactly
(\cref{app:d2q9}).
The ideal arithmetic map agrees with classical forced MRT evolution
(\cref{eq:commuting}), and one finite-width step obeys the truncation bound
(\cref{eq:fp-bound}).
The D3Q19 HIT calculations cover $10^6$ steps on the register emulator of
the digital circuit.
The D2Q9 circuits agree bit-exactly with their integer twin under
final-only measurement.

\section{Relation to existing quantum LBM and Lindblad-simulation methods}
\label{sec:sota}
Recent QLBM formulations differ in encoding, collision nonlinearity, dissipation mechanism, and the measurement contract required for multistep evolution.
Linear advection--diffusion QLBMs have moved from single-step or reinitialized formulations to multistep schemes that avoid intermediate state extraction and reinitialization \cite{wawrzyniak2025qlbm,nagel2026multistep}.
Dynamic circuits extend the same linear ADE class \cite{wawrzyniak2025dynamic}.
Those algorithms remain linear ADE solvers.
Georgescu et al.\ describe a QLBM software framework \cite{georgescu2025qlbm}.
Schalkers and M{\"o}ller showed that common amplitude and computational-basis encodings cannot make both collision and streaming unitary, forcing careful encoding choices or repeated measurement/reinitialization \cite{schalkers2024encoding}.
Carleman lifts embed nonlinear collision in a larger linear system
\cite{itani2024qalb,li2023carleman,sanavio2024carleman,sanavio2024routes,bastida2026carleman}.
Unitary quantum linear systems algorithm (QLSA) block-encoding realizations can additionally incur a recovery probability below unity.
Khan realizes an invertible finite Carleman endpoint by autonomous GKSL dynamics \cite{khan2026lindblad}.
The present algorithm acts on the nonlinear digital map $\Phi_b$ by compute-and-reset.
Wang et al.\ use a nonlinear lattice-gas ensemble with an entropy-restoring H-step and report large-scale demonstrations \cite{wang2025nonlinear}.
Tiwari et al.\ describe hardware-realizable advection--diffusion QLBM pipelines \cite{tiwari2025realizable}. Duong et al.\ replace tomography-heavy updates by denoising/projection collisions \cite{duong2026denoising}.
Liu et al.\ simulate Lindblad dynamics by trajectory averaging without postselection \cite{liu2025lindblad}.
Signed-rail amplitude damping is a CPTP treatment of isolated diagonal MRT dissipation \cite{khan2026deterministic}.

The D3Q19 backend evaluates the timestep as register operations.
A classical Markov map admits a Kraus embedding \cite{nielsen2010quantum,watrous2018tqi}.
The nonlinear digital LBM map is that channel.
Forced MRT D3Q19 DNS and LES are the nonlinear test.
D2Q9 circuits show that the same contract compiles on a second stencil (\cref{app:d2q9}).
\Cref{tab:sota-comparison} summarizes the comparison.

\begin{table}[t]
\centering
\caption{Comparison with the most directly comparable QLBM lines.
Related methods appear in \cref{sec:sota}.
Measurement contract means the intermediate measurement and reinitialization policy.
Forced HIT is forced homogeneous isotropic turbulence at the production horizons of this paper.}
\label{tab:sota-comparison}
\small
\setlength{\tabcolsep}{4pt}
\renewcommand{\arraystretch}{1.25}
\begin{tabularx}{\textwidth}{@{}>{\raggedright\arraybackslash}p{3.2cm} >{\raggedright\arraybackslash}X >{\raggedright\arraybackslash}X c c@{}}
\toprule
Method & Encoding / collision & Dissipation & Meas.\ contract & Forced HIT \\
\midrule
Signed-rail CPTP MRT \cite{khan2026deterministic}
  & digital rails; prescribed $\Lambda$
  & CPTP amplitude damping
  & local decode$^{\mathrm{a}}$ & no \\
Schalkers--M{\"o}ller \cite{schalkers2024encoding}
  & space-time / basis encodings
  & unitary focus
  & encoding-dependent & no \\
Carleman QLBM \cite{itani2024qalb,li2023carleman,bastida2026carleman}
  & Carleman / QLSA lift
  & lifted linear
  & recovery$^{\mathrm{b}}$ & no \\
Autonomous GKSL Carleman \cite{khan2026lindblad}
  & vacuum coherences; Carleman $A_K$
  & autonomous GKSL
  & encode/decode$^{\mathrm{f}}$ & no \\
Wang et al.\ ensemble \cite{wang2025nonlinear}
  & amplitude ensemble; lattice-gas + H-step
  & ensemble collision
  & end-of-run$^{\mathrm{c}}$ & decaying$^{\mathrm{e}}$ \\
Tiwari et al.\ ADE \cite{tiwari2025realizable}
  & tensor-network ADE
  & unitary ADE
  & hardware pipeline & no \\
Nagel and L{\"o}we ADE \cite{nagel2026multistep}
  & amplitude ADE
  & linear ADE
  & none mid-run$^{\mathrm{g}}$ & no \\
\midrule
\textbf{This work}
  & digital register; compute--reset ($DdQq$; D3Q19 + D2Q9)
  & fixed-point multiply + bank reset
  & final only$^{\mathrm{d}}$ & emulator$^{\mathrm{d}}$ \\
\bottomrule
\end{tabularx}

\vspace{0.35em}
{\footnotesize
$^{\mathrm{a}}$Unit success for the exposed diagonal dissipative block.
$^{\mathrm{b}}$Typical linear-combination-of-unitaries (LCU) / block-encoding or QLSA recovery overhead.
$^{\mathrm{c}}$Intermediate reinitialization avoided by ensemble design.
$^{\mathrm{d}}$Forced HIT on the computational-basis register emulator. D2Q9 is the gate-level cross-stencil check (\cref{app:d2q9}).
$^{\mathrm{e}}$Wang et al.\ report decaying HIT.
$^{\mathrm{f}}$Initial encode and final decode. Invertible linear endpoint.
$^{\mathrm{g}}$Linear ADE; no intermediate extraction or reinitialization.
}
\end{table}

\section{Classical MRT-LBM, DNS and LES}
The D3Q19 populations $\vect{f}=(f_i)_{i=0}^{18}$ satisfy the lattice
Boltzmann equation
\begin{equation}
 f_i(\vect{x}+\vect{c}_i\Delta t,t+\Delta t)
 =f_i(\vect{x},t)+\Omega_i(\vect{x},t)+\Delta t\,\Psi_i(\vect{x},t).
 \label{eq:lbe}
\end{equation}
Here $\vect{c}_i$ are the D3Q19 discrete velocities, $\Delta t$ is the
lattice timestep, $\Omega_i$ is the MRT collision increment, and
$\Psi_i$ is the Guo force increment in population space.
Let $M$ be the D3Q19 moment basis
\cite{dhumieres2002mrt}, with inverse $M^{-1}$, and let $I$ be the
$19\times 19$ identity on moment space.
The discrete velocities, weights, the numerical matrix $M$, and the
ghost rates are those tabulated by Khan et al.\
\cite{khan2026fhit}.
The diagonal relaxation matrix is $\mS=\mathrm{diag}(s_0,\ldots,s_{18})$
and the multiplier matrix is
$\Lambda=I-\mS=\mathrm{diag}(\lambda_0,\ldots,\lambda_{18})$ with
$\lambda_r=1-s_r$.
Moments are $\vect{m}=M\vect{f}$.
The raw Guo force source is $\vect{\Phi}$, written in registers in
\cref{eq:force-register}.
The MRT-corrected population source is
\begin{equation}
 \vect{\Psi}
 =M^{-1}\Bigl(I-\tfrac{\mS}{2}\Bigr)M\vect{\Phi}.
 \label{eq:mrt-force-pop}
\end{equation}
The post-collision forced moments are
\begin{equation}
 \vect{m}^{+}
 =\vect{m}^{\mathrm{eq}}
 +\Lambda(\vect{m}-\vect{m}^{\mathrm{eq}})
 +\Delta t\Bigl(I-\tfrac{\mS}{2}\Bigr)M\vect{\Phi}.
 \label{eq:mrt-forced}
\end{equation}
The complete classical timestep is
\begin{equation}
 \cC_{\Lambda}(f)
 =\cS\Biggl[
 M^{-1}\Bigl(
 \vect{m}^{\mathrm{eq}}
 +\Lambda(\vect{m}-\vect{m}^{\mathrm{eq}})
 +\Delta t\bigl(I-\tfrac{\mS}{2}\bigr)M\vect{\Phi}
 \Bigr)
 \Biggr],
 \label{eq:classical-step}
\end{equation}
where $\cS$ is periodic streaming. External body forcing uses the
discrete source of Guo et al.\ \cite{guo2002forcing}. The density and half-force velocity are
\begin{equation}
 \rho=\sum_i f_i,
 \qquad
 \rho\vect{u}
 =\sum_i f_i\vect{c}_i+\tfrac{\Delta t}{2}\vect{F},
 \label{eq:macros}
\end{equation}
with $\vect{F}$ the body-force density, and the low-Mach equilibrium form
\begin{equation}
 f_i^{\mathrm{eq}}
 =w_i\rho\Biggl[
 1+\frac{\vect{c}_i\cdot\vect{u}}{c_s^2}
 +\frac{(\vect{c}_i\cdot\vect{u})^2}{2c_s^4}
 -\frac{\vect{u}^2}{2c_s^2}
 \Biggr]
 \label{eq:feq}
\end{equation}
uses the D3Q19 weights $w_i$ and lattice sound speed $c_s=1/\sqrt{3}$.
The same equilibrium is written in that moment basis.
\Cref{tab:mrt-modes} defines the conserved, shear, and remaining nonconserved mode groups used below, together with the diagonal rates $s_r$ and multipliers $\lambda_r=1-s_r$.
The D2Q9 circuits of \cref{app:d2q9} use the moment basis of
\cite{lallemand2000theory}.

\begin{table}[t]
\centering
\caption{D3Q19 MRT diagonal collision used throughout this work
(moment-row ordering of \cite{dhumieres2002mrt},
as tabulated in \cite{khan2026fhit}).
Multipliers are $\lambda_r=1-s_r$.
The fifteen nonconserved modes enter the signed fixed-point collision;
conserved moments are left unchanged.
Open-system irreversibility enters through the subsequent bank-reset channel.}
\label{tab:mrt-modes}
\begin{tabular}{llp{0.46\linewidth}}
\toprule
Role & Indices $r$ & Rates $s_r$ \\
\midrule
Conserved ($\rho,\vect{j}$)
  & $0,3,5,7$
  & $s_r=0$ (conserved moments unchanged) \\
Shear (viscous stress)
  & $9,11,13,14,15$
  & $s_\nu=\Delta t/\tau_0$
    (DNS: fixed molecular;
     LES: local $s_\nu(\vect{x},t)$) \\
Other nonconserved
  & $1,2,4,6,8,10,12,16,17,18$
  & fixed set:
    $s_1=1.19$,
    $s_2=s_{10}=s_{12}=1.4$,
    $s_4=s_6=s_8=1.2$,
    $s_{16}=s_{17}=s_{18}=1.98$ \\
\bottomrule
\end{tabular}
\end{table}

\subsection{DNS and LES roles}
Both calculations use the forced MRT map \cref{eq:classical-step}.
That map is the MRT-Smagorinsky FHIT construction of
\cite{khan2026fhit}, now evaluated as a digital-register timestep.
They differ in the shear multiplier $\lambda_\nu$.
For DNS that multiplier is fixed by the molecular viscosity,
\begin{equation}
 \tau_0=\frac{\nu_0}{c_s^2}+\frac{\Delta t}{2},
 \qquad
 s_\nu=\frac{\Delta t}{\tau_0},
 \qquad
 \lambda_\nu=1-s_\nu.
 \label{eq:dns-lambda}
\end{equation}
LES uses the same collide--stream step on a coarser grid.
An eddy viscosity $\nu_t=(C_s\Delta)^2|S|$ is added to $\nu_0$, and the
five shear modes receive a local $\lambda_\nu(\vect{x},t)$ from
\cref{eq:dns-lambda} with $\nu_0+\nu_t$ in place of $\nu_0$.
That discrete Smagorinsky rule is the lattice form of the filtered momentum
equation
\begin{equation}
 \frac{\partial\bar u_i}{\partial t}
 +\bar u_j\frac{\partial\bar u_i}{\partial x_j}
 =-\frac{1}{\rho}\frac{\partial\bar p}{\partial x_i}
 +\nu_0\nabla^2\bar u_i
 -\frac{\partial\tau_{ij}^{\mathrm{SGS}}}{\partial x_j}
 +\frac{\bar F_i}{\rho},
 \label{eq:filtered-nse}
\end{equation}
with $\tau_{ij}^{\mathrm{SGS},d}=-2\nu_t\bar S_{ij}$ the traceless
subgrid stress, Smagorinsky constant $C_s$, filter width $\Delta$, and
$|S|$ the strain-rate magnitude.
Here $\vect{F}$ is the body-force density of \cref{eq:macros}, so
$\vect{F}/\rho$ is the acceleration.
The strain $|S|$ is reconstructed from nonequilibrium moments following
Yu et al.\ \cite{yu2006mrtles}, as in \cite{khan2026fhit}.
\Cref{sec:les-multipliers} forms the digital shear multiplier from that $|S|$.
Each quantum-register simulation has a matched classical reference with
identical lattice, initialization, forcing, rates, duration, and
diagnostics.

\section{Digital compute--reset architecture}
\label{sec:digital-collision}
Let $\cR$ be a finite set, the alphabet of a digital register, and let
$\Phi:\cR\to\cR$ be any deterministic map that admits a reversible clean
evaluation into a second register.
At timestep $t$, write $A$ for the source bank holding $\ket{x}$, $B$ for a
clean destination bank, and $W$ for arithmetic workspace.
A reversible circuit implements
\begin{equation}
 U_{\Phi}:
 \ket{x}_{A}\ket{0}_{B}\ket{0}_{W}
 \;\longrightarrow\;
 \ket{x}_{A}\ket{\Phi(x)}_{B}\ket{0}_{W} .
 \label{eq:u-step}
\end{equation}
The map is evaluated by compute--copy--uncompute.
$\Phi(x)$ is written into workspace, copied into the destination bank, and
the workspace is uncomputed to $\ket{0}$.
The obsolete source bank is then erased by
\begin{equation}
 \mathrm{Reset}(A),
 \label{eq:reset-source}
\end{equation}
giving $\ket{0}_{A}\ket{\Phi(x)}_{B}\ket{0}_{W}$.
For the following timestep the populated destination bank is the new
source and the reset source bank is the new destination.
This is a change of register roles.
The complete algorithm is
\begin{equation}
 \boxed{
 \text{prepare }x_0
 \;\rightarrow\;
 \prod_{t=0}^{T-1}
 \left[U_{\Phi}
 \rightarrow\mathrm{Reset}(A)\right]
 \;\rightarrow\;
 \text{final measurement}
 }.
 \label{eq:boxed-algo}
\end{equation}

\begin{proposition}[Compute--reset CPTP channel]
\label{prop:cptp-reset}
Let $\Phi:\cR\to\cR$ be any deterministic map on a finite set $\cR$,
let $\cH_A\cong\cH_B\cong\mathbb{C}^{|\cR|}$ be the source and
destination banks, and let
\begin{equation}
 V:\cH_A\to\cH_A\otimes\cH_B,
 \qquad
 V\ket{x}=\ket{x}\ket{\Phi(x)},
 \label{eq:two-bank-isometry}
\end{equation}
be the reversible two-bank compute step \cref{eq:u-step}.
Then $V$ is an isometry, and the physical step
\emph{``apply $V$, then discard bank $A$''} is the channel
\begin{equation}
 \cE_{\Phi}(\rho)
 =\Tr_A\!\bigl[V\rho V^\dagger\bigr]
 =\sum_x \rho_{xx}\,\ketbra{\Phi(x)}{\Phi(x)},
 \label{eq:bank-reset-channel}
\end{equation}
with Kraus operators $K_x=\ketbra{\Phi(x)}{x}$.
Consequently $\cE_{\Phi}$ is CPTP, it acts as $\Phi$ on every
computational-basis input,
$\cE_{\Phi}(\ketbra{x}{x})=\ketbra{\Phi(x)}{\Phi(x)}$, and on an
arbitrary input it destroys all coherences in the computational basis.
\end{proposition}

\begin{proof}
$V^\dagger V=\sum_{x,x'}\braket{x'}{x}\braket{\Phi(x')}{\Phi(x)}
\ketbra{x'}{x}=\sum_x\ketbra{x}{x}=I$, so $V$ is an isometry and
$\rho\mapsto V\rho V^\dagger$ is a reversible dilation.
$V$ is an isometry that can be extended to a unitary on a sufficiently
enlarged register space. No assumption of injectivity of $\Phi$ is needed, because the
source register itself carries the label $x$.
Writing $\rho=\sum_{x,x'}\rho_{xx'}\ketbra{x}{x'}$,
\[
 V\rho V^\dagger
 =\sum_{x,x'}\rho_{xx'}\,
 \ketbra{x}{x'}\otimes\ketbra{\Phi(x)}{\Phi(x')},
\]
and tracing out bank $A$ retains only $x=x'$, which gives
\cref{eq:bank-reset-channel}. Since discarding a subsystem is the partial
trace, and unconditional $\mathrm{Reset}(A)$ followed by reuse of $A$ is
physically the same operation up to a local state preparation on $A$, the
reset implements exactly this channel.
Identifying $K_x=\ketbra{\Phi(x)}{x}$ gives
$\sum_x K_x^\dagger K_x=\sum_x\ketbra{x}{x}=I$, the Kraus (operator-sum)
form of a CPTP map \cite{nielsen2010quantum,watrous2018tqi}.
Setting $\rho=\ketbra{x}{x}$ yields
$\ketbra{\Phi(x)}{\Phi(x)}$, while for $x\neq x'$ the off-diagonal
terms are annihilated, so $\cE_{\Phi}$ factors through complete
computational-basis dephasing and is entanglement-breaking.
The reset exports the discarded bank-$A$ information into the
environment. That is the dissipative part of the construction, and it is
the only irreversible element of the timestep.
\end{proof}

The Kraus form of compute-and-discard is standard
\cite{nielsen2010quantum,watrous2018tqi}.
A complete nonlinear digital LBM timestep sits in that contract, so
timesteps compose on the register state.
If $\Phi_b$ is any deterministic fixed-point LBM timestep map on a
$DdQq$ population register, and if it admits the evaluation
\cref{eq:u-step}, then \cref{eq:bank-reset-channel} implements $\Phi_b$
on computational-basis inputs.
Forced D3Q19 MRT DNS, D3Q19 Smagorinsky LES, and the D2Q9 circuits of
\cref{app:d2q9} are instances of that LBM map.
BGK and TRT occupy the same slot, with one or two relaxation rates.
Bank reset dephases in the computational basis.

The logical circuit measures at the final readout.
A logical $\mathrm{Reset}$ may be compiled as mid-circuit measurement
and a corrective $X$, with the outcome discarded.

From here on, $\cR$ is the digital LBM population space.
Let $\Lat$ be a periodic lattice in $d$ dimensions and
$\mathcal{C}=\{c_0,\ldots,c_{Q-1}\}$ a $DdQq$ velocity stencil.

\subsection{Signed MRT multiply and D3Q19 consistency}
\label{sec:mrt-multiply}
A digital-register QLBM stores one signed fixed-point population word per
velocity at each site.
Write a signed format $(w,f)$, where $w$ is the two's-complement word width
including the sign bit and $f$ is the number of fractional bits.
Its representable set is
\begin{equation}
 \mathbb{F}_{w,f}
 =\bigl\{k\,2^{-f}:\;
 -2^{\,w-1}\leq k\leq 2^{\,w-1}-1,\;k\in\mathbb{Z}\bigr\},
 \label{eq:fp-grid}
\end{equation}
and $\mathrm{sat}_{w}$ clamps an integer to $[-2^{w-1},2^{w-1}-1]$.
Throughout, $\mathrm{round}(\cdot)$ denotes round-half-away-from-zero, which
is symmetric under sign reversal.
In the main text $b$ always denotes the \emph{fractional} width $f$ of the
production registers, with integer widths fixed per register class.
Symmetry of the rounding rule is required because overrelaxation drives
$\lambda$ negative, and a floor-based rule would bias every negative
product low by one unit in the last place.

The classical state space is $\cX=(\mathbb{R}^{Q})^{\Lat}$.
The register state space is the finite set
$\cR=(\mathbb{F}_{w,f}^{\,Q})^{\Lat}$ of population fields whose entries
lie on the fixed-point grid \cref{eq:fp-grid}.
The decoding map
\[
 \cD:\cR\to\cX,
 \qquad
 (\cD F)_i(\vect{x})=F_i(\vect{x}),
\]
is the injective inclusion that reads each register word as the rational
number it represents.
The encoding map $\Pi:\cX\to\cR$ is the elementwise quantizer of the
persistent population format $(w,f)$,
\begin{equation}
 \Pi(x)=2^{-f}\,\mathrm{sat}_{w}\!\bigl(\mathrm{round}(2^{f}x)\bigr),
 \label{eq:quantizer}
\end{equation}
so that $\Pi\circ\cD=\mathrm{id}_{\cR}$.
Work registers (moments, features, sources, reciprocals, accumulators)
may use other integer widths.
Their fractional widths are at least $b$.
The persistent state space $\cR$ is the population field.
Let $\Phi_b:\cR\to\cR$ be a deterministic fixed-point LBM timestep on this
space.
Given a multiplier rule $G$, write $\cQ_{G,b}:\cR\to\cR$ for the finite-width
register update.
Write $\cQ_{G,\infty}:\cX\to\cX$ for the same mathematical timestep in exact
real arithmetic, with exact $1/\rho$ and $\sqrt{\cdot}$ in place of finite
Newton iteration, and with quantization and saturation removed.
The classical update $\cC_G:\cX\to\cX$ is the matching exact-real discrete
LBM timestep.
Both evaluate their multipliers from their own state, $\Lambda=G(\cdot)$.
The production calculations use $d=3$, $Q=19$, and the forced MRT map
of \cref{eq:classical-step}.
MRT mode relaxation is implemented by signed fixed-point multiplication.
For each nonconserved mode, write the moment deviation $\delta m$ and
the multiplier $\lambda=1-s$ as scaled two's-complement integers, with
$f_m$ and $f_\lambda$ the fractional widths of those registers,
\begin{equation}
 D=\mathrm{round}(2^{f_m}\delta m),\qquad
 L=\mathrm{round}(2^{f_\lambda}\lambda),\qquad
 \lambda=1-s,
 \label{eq:fp-integers}
\end{equation}
and compute
\begin{equation}
 D'
 =\mathrm{sat}_{w_m}\!\left[
   \mathrm{round}\!\left(\frac{D\,L}{2^{f_\lambda}}\right)
 \right].
 \label{eq:fp-multiply}
\end{equation}
The quotient in \cref{eq:fp-multiply} is evaluated from the exact integer
product $D\,L$, so the rounding is that of the infinitely precise product.
When $s>1$, one has $\lambda=1-s<0$, so $L$ is simply a negative integer.
The same signed-multiply contract implements overrelaxation.
For forced D3Q19 MRT the digital timestep is the composition
\begin{equation}
 \Phi_b
 =\mathrm{stream}\circ\mathrm{inverse\,MRT}\circ\mathrm{forcing}
 \circ\mathrm{collision}\circ\mathrm{equilibrium}\circ\mathrm{MRT},
 \label{eq:phi-b}
\end{equation}
with every arithmetic stage obeying an explicit integer width, fractional
width, rounding rule, saturation rule, and (for $1/\rho$ and $\sqrt{\cdot}$)
$J_{\mathrm{r}}=J_{\mathrm{s}}=6$ Newton iterations, sufficient to reach the register rounding.
That composition is the D3Q19 instance of $\cQ_{G,b}$.
The production backend evaluates it as register operations.
In that instantiation $\cC_G$ is the exact-real discrete forced MRT map of
\cref{eq:classical-step}, and floating-point arithmetic is the numerical
reference used to evaluate it.

In exact real arithmetic $\cQ_{G,\infty}$ recovers
the same macros $(\rho,\vect{u})$, moments $\vect{m}$, and equilibrium
$\vect{m}^{\mathrm{eq}}$ as the classical kernels, with
$\Lambda=G(f)$ constant for DNS and given by the Smagorinsky closure for
LES.
Signed multiplication yields $\Lambda(\vect{m}-\vect{m}^{\mathrm{eq}})$,
and register evaluation of the discrete force source reproduces
$\Delta t\bigl(I-\mS/2\bigr)M\vect{\Phi}$.
The same inverse moment transform and periodic streaming follow, so the two
exact-real maps coincide,
\begin{equation}
 \cQ_{G,\infty}=\cC_{G}.
 \label{eq:commuting}
\end{equation}
For a register input $F\in\cR$ this is
$\cQ_{G,\infty}(\cD F)=\cC_{G}(\cD F)$,
and the identity iterates for every number of timesteps.
For state-dependent LES, both algorithms evaluate $\Lambda_t=G(f_t)$ from
their current population states at every timestep.

If all intermediate arithmetic quantities remain within their prescribed
ranges, if every arithmetic register has
fractional width at least $b$, and if the constituent maps and the
LES rule $G$ are locally Lipschitz on the states considered, one timestep
of the finite-width register map satisfies
\begin{equation}
 \bigl\|
 \cD\cQ_{G,b}(F)-\cQ_{G,\infty}(\cD F)
 \bigr\|
 =
 \bigl\|
 \cD\cQ_{G,b}(F)-\cC_{G}(\cD F)
 \bigr\|
 \leq
 C_{\mathrm{fp}}\,2^{-b},
 \label{eq:fp-bound}
\end{equation}
with $C_{\mathrm{fp}}$ independent of $b$.
This is a one-step arithmetic bound.
Forced HIT is chaotic, so two solvers that differ by
$\mathcal{O}(2^{-b})$ at one step separate and decorrelate pointwise
regardless of $b$.
The width survey of \cref{tab:fixedpoint-robustness} shows collapse at $b=16$.
Widths $b=24$, $32$, and $40$ remain in the sustained kinetic-energy band.
Pre-chaotic kinetic-energy error falls with $b$, as expected from
\cref{eq:fp-bound}.
Long $32^3$ trajectories still mix under chaos, so pairwise final-field
$L^2$ mixes across widths.
Long-horizon fidelity is therefore assessed statistically through spectra,
budgets, and the post-transient moments of \cref{tab:post-transient}.

\section{Forced D3Q19 instantiation}
\label{sec:digital-algo}
The persistent state is a field of signed fixed-point D3Q19 population registers
\begin{equation}
 |F_t\rangle
 =\bigotimes_{\vect{x}}\bigotimes_{i=0}^{18}
 |f_i(\vect{x},t)\rangle.
 \label{eq:population-register}
\end{equation}
This section instantiates the digital LBM map of \cref{sec:digital-collision} on forced D3Q19 MRT.
The production solver is that register map.
It advances $F_t$ through the
stages of \cref{eq:phi-b} at the declared register boundaries.
Finite bit-width surveys appear in
\cref{tab:fixedpoint-robustness}
and \cref{sec:fixedpoint-results,app:fixedpoint}.
Work registers store density, momentum, reciprocal density, velocity,
moments, equilibrium features, source moments, and strain quantities.
Each numerical map corresponds to reversible arithmetic, unconditional
reset, or a permutation.
\Cref{fig:qlbm-register-dns-les} summarizes the forced D3Q19 register
timestep of \cref{tab:register-operations} and the DNS/LES construction
of its control registers.

\begin{table}[t]
\centering
\caption{Register operations in one digital QLBM timestep.
Every row has a numerical basis-register implementation and a circuit-level interpretation.
The logical circuit measures at the final readout.
Scheduled diagnostics are emulator-only.}
\label{tab:register-operations}
\begin{tabular}{p{0.22\linewidth}p{0.34\linewidth}p{0.36\linewidth}}
\toprule
Stage & Numerical map & Circuit interpretation \\
\midrule
Macroscopic recovery & $F\mapsto(\rho,\vect{j},\vect{u})$ & Reversible adder trees, fixed-point reciprocal, products \\
Moment transform & $F\mapsto M F$ & Fixed sparse linear circuits \\
Equilibrium & $(\rho,\vect{u})\mapsto m^{\mathrm{eq}}$ & Quadratic feature registers and constant multipliers \\
Nonequilibrium moments & $\delta m=m-m^{\mathrm{eq}}$ ($+\tfrac{\Delta t}{2}\Phi_m$ for LES strain) & Signed fixed-point subtraction \\
Digital collision & $\delta m\mapsto\lambda\odot\delta m$ & Signed two's-complement fixed-point multiply \\
Body forcing & $(\vect{u},\vect{F})\mapsto\vect{\Phi}\mapsto\Delta t(I-\mS/2)M\vect{\Phi}$ & Reversible contractions and source-register addition \\
Inverse transform & $m^{+}\mapsto F^{+}$ & Fixed sparse inverse map \\
Streaming & $F^{+}(\vect{x})\mapsto F(\vect{x}+\vect{c}_i)$ & Periodic modular permutation \\
Clean compute & $(F_t,0_B,0_W)\mapsto(F_t,\Phi_b(F_t),0_W)$ & Compute--copy--uncompute ($U_{\Phi}$) \\
Bank erase & $F_t\mapsto 0$ on bank $A$ & Unconditional $\mathrm{Reset}(A)$ (open-system CPTP) \\
Bank reuse & $(0_A,F_{t+1,B})$ & Use populated bank $B$ as the next source and reset bank $A$ as the next destination \\
\bottomrule
\end{tabular}
\end{table}

\begin{figure}[t]
\centering
\refstepcounter{algorithm}
\label{alg:digital-timestep}
\fbox{\parbox{0.94\linewidth}{
\textbf{Algorithm~\thealgorithm.} Forced D3Q19 instance of the two-bank compute--reset timestep.\\[0.35em]
\textbf{Input:} source $S_t=|F_t\rangle$, destination $D_t=|0\rangle$, workspace $W=|0\rangle$, and DNS rates or LES closure $G$.\\
\textbf{Output:} $S_t=|0\rangle$, $D_t=|F_{t+1}\rangle$, $W=|0\rangle$.
\begin{enumerate}[leftmargin=1.4em]
\item Recover $(\rho,\vect{u})$ and $\vect{m}=M\vect{f}$ by reversible fixed-point arithmetic, then evaluate $\vect{m}^{\mathrm{eq}}$ and $\delta\vect{m}$.
\item Form $\Lambda=G(f)$ (fixed for DNS, Smagorinsky $\lambda_\nu(\vect{x},t)$ for LES).
\item For each nonconserved mode, compute $D_r'=\mathrm{sat}(\mathrm{round}(D_r L_r/2^{f_\lambda}))$ by signed fixed-point multiply (\cref{eq:fp-multiply}).
\item Reconstruct moments, apply body forcing and inverse $M$, and stream the result into the clean destination bank $D_t$.
\item Uncompute all arithmetic workspace ($W\to|0\rangle$).
\item Apply $\mathrm{Reset}(S_t)$ to the obsolete source bank.
\item At the next timestep use $D_t$ as the source and the reset $S_t$ as the destination. For LES, recompute $\lambda_\nu$ from the new register state without classical reinjection.
\end{enumerate}
}}
\end{figure}

\subsection{Macroscopic, equilibrium, collision, forcing, and streaming}
Reversible adder trees form $\rho$ and $\vect{j}$. A fixed-point
reciprocal and products form $\vect{u}$. The fixed matrix $M$ is
implemented by sparse signed additions. Quadratic work registers
generate the D3Q19 equilibrium moments, after which
\begin{equation}
 |\vect{m}\rangle|\vect{m}^{\mathrm{eq}}\rangle|0\rangle
 \mapsto
 |\vect{m}\rangle|\vect{m}^{\mathrm{eq}}\rangle
 |\delta\vect{m}\rangle.
 \label{eq:noneq-map}
\end{equation}
The fifteen nonconserved entries enter the signed multiply
\cref{eq:fp-multiply}. The discrete force term is evaluated in registers as
\begin{equation}
 \Phi_i
 =w_i\left[
  \frac{d_i^{F}-d^{uF}}{c_s^2}
 +\frac{d_i^{u}d_i^{F}}{c_s^4}
 \right],
 \label{eq:force-register}
\end{equation}
with $d_i^{F}=\vect{c}_i\cdot\vect{F}$, $d_i^{u}=\vect{c}_i\cdot\vect{u}$, and $d^{uF}=\vect{u}\cdot\vect{F}$.
The corresponding moment-space correction
$\Delta t\bigl(I-\mS/2\bigr)M\vect{\Phi}$ is added before the inverse transform.
Periodic streaming is a fixed permutation of the explicit population words,
\begin{equation}
 F_{t+1,i}(\vect{x}+\vect{c}_i\Delta t)
 =F_i^{+}(\vect{x}),
 \label{eq:quantum-stream}
\end{equation}
equivalently
\[
 \bigotimes_{\vect{x},i}
 |f_i^{+}(\vect{x})\rangle
 \longmapsto
 \bigotimes_{\vect{x},i}
 |f_i^{+}(\vect{x}-\vect{c}_i\Delta t)\rangle.
\]
Reverse arithmetic clears work registers. Only the obsolete state bank
is erased by $\mathrm{Reset}(A)$.

\section{State-dependent LES-controlled digital multipliers}
\label{sec:les-multipliers}
The force-corrected nonequilibrium used for strain is
$\widetilde{\delta\vect{m}}=\vect{m}-\vect{m}^{\mathrm{eq}}+\tfrac{\Delta t}{2}\vect{\Phi}_m$
with $\vect{\Phi}_m=M\vect{\Phi}$
\cite{premnath2009forcing,chai2012mrtstrain}.
The D2Q9 gate-level circuits in \cref{app:d2q9} use a lattice-specific reduction of the same MRT stress relation.
With strain-rate tensor $S_{\alpha\beta}$ (distinct from the MRT matrix $\mS$) and magnitude $|S|=\sqrt{2S_{\alpha\beta}S_{\alpha\beta}}$ evaluated by fixed-point square-root iteration,
\begin{equation}
 \nu_t=(C_s\Delta)^2|S|,
 \qquad
 \nu_{\mathrm{eff}}=\nu_0+\nu_t,
 \label{eq:eddy-viscosity}
\end{equation}
the local shear multiplier is
\begin{equation}
 \tau_{\mathrm{eff}}=\frac{\nu_{\mathrm{eff}}}{c_s^2}+\frac{\Delta t}{2},
 \qquad
 s_\nu(\vect{x},t)=\frac{\Delta t}{\tau_{\mathrm{eff}}},
 \qquad
 \lambda_\nu(\vect{x},t)=1-s_\nu(\vect{x},t).
 \label{eq:effective-relaxation}
\end{equation}
The five shear modes receive $s_\nu(\vect{x},t)$. Other nonconserved rates retain their prescribed MRT values.
The strain reconstruction contains $s_\nu$.
The D3Q19 digital map removes that circular dependence by evaluating strain with the molecular shear rate $1/\tau_0$ of \cref{eq:dns-lambda}, then inserting the resulting $s_\nu$ into the multiply.
The persistent digital state remains the population field $F_t$, so $\Lambda_t=G(f_t)$.
The D2Q9 circuits of \cref{app:d2q9} instead keep a persistent shear-rate seed and one Picard step.
The integer $L=\mathrm{round}(2^{f_\lambda}\lambda_\nu)$ then enters
\cref{eq:fp-multiply}.
Since $\nu_t\geq0$ and $\tau_0>\Delta t/2$, one has $-1<\lambda_\nu<1$,
so every local multiplier remains admissible for the signed multiply
(including overrelaxation when $\lambda_\nu<0$).

The LES pipeline therefore applies a state-dependent sequence of digital
multiplies, with $\lambda_\nu(\vect{x},t)$ computed from the evolving
register state at every node and timestep, without intermediate
measurement or classical state reinjection.

\begin{figure}[!t]
 \centering
 \includegraphics[width=\linewidth]{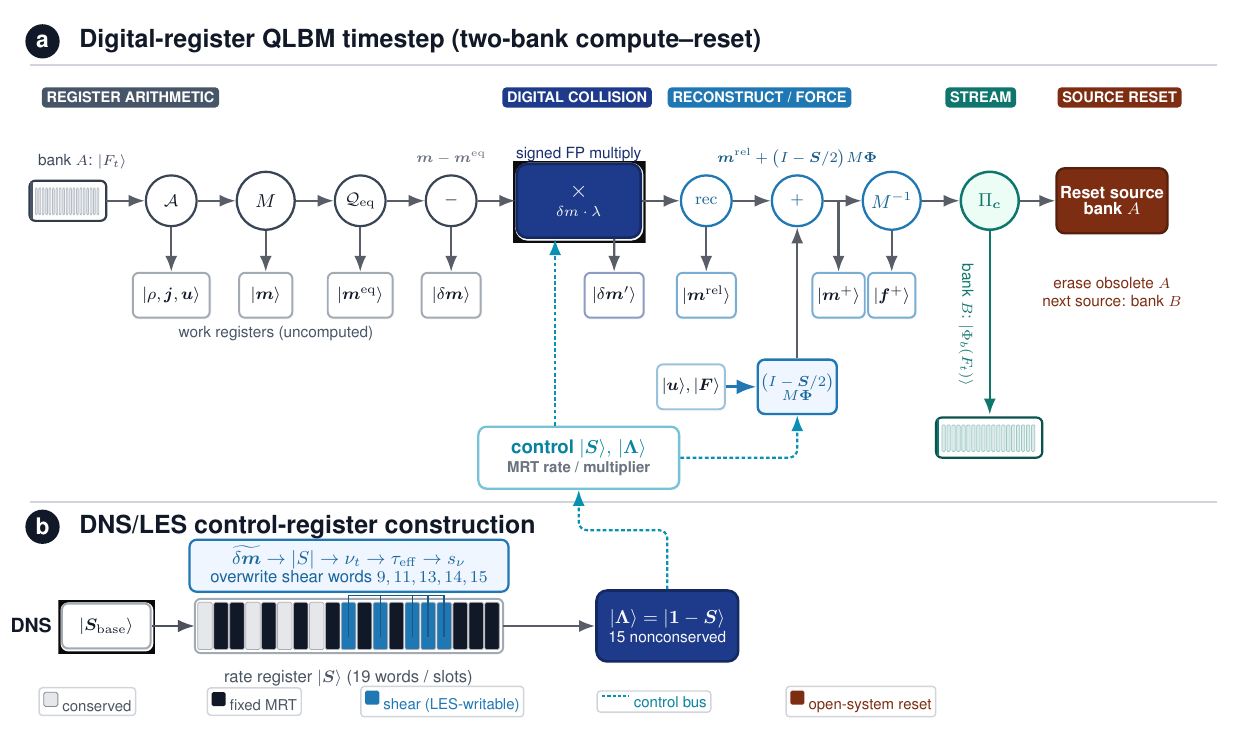}
 \caption{Digital-register forced D3Q19 QLBM timestep and DNS/LES control registers.
 (a)~Register pipeline from source bank $A=|F_t\rangle$ to destination bank $B=|F_{t+1}\rangle$.
 The stages are macroscopic recovery, moment transform, equilibrium arithmetic, nonequilibrium moments,
 signed fixed-point collision $\delta m\mapsto\lambda\odot\delta m$, reconstruction and body forcing
 $\Delta t(I-\mS/2)M\vect{\Phi}$, inverse transform, streaming into bank $B$,
 then $\mathrm{Reset}(A)$. Bank $B$ supplies the source state for the next timestep
 (cf.~\cref{tab:register-operations} and Algorithm~\ref{alg:digital-timestep}).
 Dashed cyan lines carry the MRT rate and multiplier registers $|\bm S\rangle$ and $|\bm\Lambda\rangle$.
 (b)~DNS loads fixed $|\bm S_{\mathrm{base}}\rangle$.
 LES overwrites shear words $9,11,13,14,15$ from local strain.
 Fifteen nonconserved multipliers $|\bm\Lambda\rangle=|\bm 1-\bm S\rangle$ feed the digital multiply in~(a).}
 \label{fig:qlbm-register-dns-les}
\end{figure}

\section{Emulator and numerical setup}
A direct statevector of the full lattice would grow exponentially with register count.
The production D3Q19 solver is the register implementation of
Algorithm~\ref{alg:digital-timestep}.
It evolves the computational-basis
population words through the same arithmetic stages as the logical
circuit.
Scheduled turbulence diagnostics are emulator-only.
The logical circuit measures at the final readout.

\subsection{Register-level and gate-level validation}
\label{sec:emulator-fidelity}
The two validation layers test the same compute--reset contract at different
resolutions.
The D3Q19 register backend is faithful at register boundaries.
Every stage of \cref{eq:phi-b} writes its result through the quantizer of
its register class, so stored values lie on the grid of \cref{eq:fp-grid}
with the declared integer and fractional widths.
Inside a stage the accumulation is carried in double precision and quantized
on output.
The MRT relaxation multiply of \cref{eq:fp-multiply} is treated differently.
It is evaluated in exact integer arithmetic on the scaled integers $D$ and
$L$, using a split product that remains inside $64$-bit integers even at
$b=40$, where the direct product would require $87$ bits.
That stage is treated exactly because a floor-based shift would bias every
negative product low by one unit in the last place, and overrelaxation makes negative
multipliers the normal case.
The long-horizon D3Q19 calculations exercise that finite-width map over $10^{6}$ steps.
The $8^3$ one-step test of \cref{tab:verification} measures the left-hand side of \cref{eq:fp-bound} against the parent floating-point MRT kernels at production width $b=24$.
Gate-level integer datapaths are demonstrated separately and bit-exactly on
D2Q9, which is a different stencil of the same architecture
(\cref{app:d2q9}).

\begin{table}[t]
\centering
\caption{Matched production configurations for the forced-HIT DNS and LES calculations
($dx=dt=1$, $Re_0=u_0N/\nu_0$).
$A$ is the body-force amplitude.
$\Delta_{\mathrm{diag}}$ is the diagnostic sampling interval in lattice steps.
QLBM-register production calculations use the digital fixed-point compute--reset backend.
DNS status uses $k_{\max}\eta\geq1$.}
\label{tab:numerical-setup}
\begin{tabular}{lcccccccc}
\toprule
Case & Grid & $\tau_0$ & $\nu_0$ & $u_0$ & $A$ & $C_s$ & Steps & $\Delta_{\mathrm{diag}}$ \\
\midrule
Classical MRT DNS & $320^3$ & 0.5064 & $2.133\times10^{-3}$ & 0.1 & $5\times10^{-6}$ & n/a & $10^{6}$ & $10^{3}$ \\
Digital QLBM DNS & $320^3$ & 0.5064 & $2.133\times10^{-3}$ & 0.1 & $5\times10^{-6}$ & n/a & $10^{6}$ & $10^{3}$ \\
Classical MRT LES & $64^3$ & 0.50128 & $4.267\times10^{-4}$ & 0.1 & $5\times10^{-6}$ & 0.13 & $10^{6}$ & $10^{3}$ \\
Digital QLBM LES & $64^3$ & 0.50128 & $4.267\times10^{-4}$ & 0.1 & $5\times10^{-6}$ & 0.13 & $8.78\times10^{5}$ & $10^{3}$ \\
\bottomrule
\end{tabular}
\end{table}

The production parameters are those of \cref{tab:numerical-setup}. Both the $320^3$ DNS candidate and the $64^3$ LES case share $Re_0=u_0N/\nu_0=15000$ with $\nu_0=(\tau_0-1/2)/3$, Taylor--Green initialization $u_0=0.1$, and forcing amplitude $5\times10^{-6}$.
DNS and classical LES run for $10^{6}$ lattice steps; digital LES for $8.78\times10^{5}$ (\cref{tab:numerical-setup}). All fields are reported in lattice units ($dx=dt=1$). Integration length and diagnostic sampling use lattice steps, while time histories are plotted in eddy turnovers $t/T_{\mathrm{eddy}}$ with $T_{\mathrm{eddy}}=u_{\mathrm{rms}}^{2}/\varepsilon_{\mathrm{tot}}$.
On DNS, $\varepsilon_{\mathrm{tot}}$ reduces to the molecular dissipation $\varepsilon_{\mathrm{mol}}$.
Spectra, structure functions, and increment flatness are recorded for
steps $\geq3\times10^{5}$. Time averages of those diagnostics and the
post-transient statistics in \cref{tab:post-transient,tab:dns-resolution}
use that common window.

Turbulence diagnostics are evaluated from three-dimensional velocity spectra on the periodic box of side $L_{\mathrm{box}}=N\Delta x$.
Shell wavenumbers use the standard Fourier spacing
\begin{equation}
 \Delta k=\frac{2\pi}{L_{\mathrm{box}}},
 \qquad
 k_n=n\Delta k,
 \qquad
 k_{\max}=\frac{\pi}{\Delta x},
 \label{eq:k-grid}
\end{equation}
with $n=0,1,\ldots,N/2$.
Plotted spectra use the radial shells $E(k)$.
Kinetic energy is the shell sum, while molecular dissipation uses the unbinned three-dimensional Fourier sum,
\begin{equation}
 \begin{aligned}
  K&=\sum_{k>0}E(k)\,\Delta k,
  &
  \varepsilon_{\mathrm{mol}}&=2\nu_0\sum_{\mathrm{modes}}\lvert\vect{k}\rvert^{2} e_{\mathrm{mode}},
  \\
  u_{\mathrm{rms}}&=\sqrt{\tfrac{2}{3}K},
  &
  \lambda_T&=\sqrt{15\nu_0\,u_{\mathrm{rms}}^{2}/\varepsilon_{\mathrm{mol}}},
  \\
  \mathrm{Re}_{\lambda}&=u_{\mathrm{rms}}\lambda_T/\nu_0,
  &
  \eta&=(\nu_0^{3}/\varepsilon_{\mathrm{mol}})^{1/4}.
 \end{aligned}
 \label{eq:hit-scales}
\end{equation}
The Kolmogorov and Taylor scales use this molecular spectral dissipation.
Compensated spectra are plotted as
$E(k)/(\varepsilon_{\mathrm{mol}}^{2/3}k^{-5/3})$
and compared with the Kolmogorov constant $C_K=1.5$.
The spectral integral length used in the plotted reference is
\begin{equation}
 L_{11}
 =\frac{3\pi}{4K}\sum_{k>0}\frac{E(k)}{k}\,\Delta k,
 \label{eq:pope-L}
\end{equation}
the integral scale entering the Pope shape, distinct from the box side
$L_{\mathrm{box}}$ and from Pope's large-eddy scale
$K^{3/2}/\varepsilon$.
As a reference shape we use a Pope-form spectrum with that measured
$L_{11}$ \cite{pope2000turbulent}
\begin{equation}
 \begin{aligned}
  E(k)&=C_K\varepsilon_{\mathrm{mol}}^{2/3}k^{-5/3}\,f_L(k L_{11})\,f_\eta(k\eta),
  \\
  f_L(k L_{11})&=\Biggl(\frac{k L_{11}}{\sqrt{(k L_{11})^{2}+c_L}}\Biggr)^{11/3},
  \\
  f_\eta(k\eta)
  &=\exp\!\Bigl(
    -\beta\bigl[((k\eta)^{4}+c_\eta^{4})^{1/4}-c_\eta\bigr]
  \Bigr),
 \end{aligned}
 \label{eq:pope}
\end{equation}
with $C_K=1.5$, $c_L=6.78$, $c_\eta=0.40$, and $\beta=5.2$ as in the
production diagnostics.
Pope comparisons are shown in Kolmogorov units
$E(k)/(\nu_0^{5}\varepsilon_{\mathrm{mol}})^{1/4}$ versus $k\eta$.
The model curve in \cref{fig:pope} is a classical DNS reference using those
DNS scales.
LES spectra appear on the same axes as a visual comparison.
Longitudinal increment flatness is
\begin{equation}
 F_L(r)
 =\frac{\langle[\delta_r u_\parallel]^{4}\rangle}
 {\langle[\delta_r u_\parallel]^{2}\rangle^{2}},
 \label{eq:flatness}
\end{equation}
averaged over the three Cartesian directions, with Gaussian reference
$F_L=3$.
Longitudinal structure functions
$S_p(r)=\langle|\delta_r u_\parallel|^{p}\rangle$ are analyzed in
extended self-similarity (ESS) form
\begin{equation}
 S_p(r)\sim\bigl[S_3(r)\bigr]^{\xi_p},
 \label{eq:ess}
\end{equation}
with scaling anomalies $\xi_p-p/3$ compared with the She--Leveque
\cite{she1994universal} and Benzi \cite{benzi1993ess} references.
Both fine-grid solvers satisfy $k_{\max}\eta\geq1$ (\cref{tab:dns-resolution}).
For LES energy-budget residuals we report
$\varepsilon_{\mathrm{tot}}=\varepsilon_{\nu}+\varepsilon_{\mathrm{SGS}}$ with
$\varepsilon_{\mathrm{SGS}}=\langle\nu_t|S|^{2}\rangle$
and $\varepsilon_{\nu}=2\nu_0\langle S_{\alpha\beta}S_{\alpha\beta}\rangle$.
On DNS, $\varepsilon_{\mathrm{SGS}}=0$ and $\varepsilon_{\mathrm{tot}}=\varepsilon_{\nu}$.
Kinetic energy $K$ is per unit mass, so the forcing power is
$P_F=\langle(\vect{F}/\rho)\cdot\vect{u}\rangle$.
Production uses $\rho_0=1$, and therefore $P_F=\langle\vect{F}\cdot\vect{u}\rangle$
in lattice units.
The residual is
$R_\varepsilon=\lvert\varepsilon_{\mathrm{tot}}/(-dK/dt+P_F)-1\rvert$.
\begin{table}[t]
\centering
\caption{Post-transient DNS resolution assessment ($dx=1$, $\nu_0=2.133\times10^{-3}$, $k_{\max}=\pi$).
Both fine-grid solvers satisfy $k_{\max}\eta\geq1$ and are retained as DNS.
Dissipation is the molecular spectral value $\varepsilon_{\mathrm{mol}}$ of \cref{eq:hit-scales}.
$\mathrm{Re}_\lambda=u_{\mathrm{rms}}\lambda_T/\nu_0$.}
\label{tab:dns-resolution}
\begin{tabular}{lcccccc}
\toprule
Case & $u_{\mathrm{rms}}$ & $\varepsilon_{\mathrm{mol}}$ & $\lambda_T$ & $\mathrm{Re}_\lambda$ & $\eta$ & $k_{\max}\eta$ \\
\midrule
Classical MRT DNS & $3.27\times10^{-2}$ & $1.64\times10^{-7}$ & $14.42$ & $221$ & $0.493$ & $1.55$ \\
Digital QLBM DNS & $3.20\times10^{-2}$ & $1.60\times10^{-7}$ & $14.34$ & $215$ & $0.497$ & $1.56$ \\
\bottomrule
\end{tabular}
\end{table}

\section{Results}
\subsection{Algorithmic verification}
\begin{table}[t]
\centering
\caption{Algorithmic verification.
One-step errors are relative $L^2$ population errors versus the parent floating-point MRT kernels
on an $8^3$ lattice using the digital fixed-point backend (nominal fractional width $b=24$).
Long-time pairwise diagnostics use the matched production calculations of \cref{tab:numerical-setup}.
The TKE-history $L^2$ difference measures trajectory decorrelation.
Statistical fidelity is judged from spectra, budgets, and \cref{tab:post-transient}.
Mass drift is $\max_t |m(t)/m(0)-1|$ for the QLBM-register calculations.
Classical float64 drift remains $\mathcal{O}(10^{-11})$.}
\label{tab:verification}
\begin{tabular}{lcc}
\toprule
Quantity & DNS & LES \\
\midrule
One-step population error (vs float parent) & $<5\times10^{-4}$ & $<5\times10^{-3}$ \\
Time-averaged spectrum relative difference & $5.45\times10^{-2}$ & $3.05\times10^{-2}$ \\
TKE-history pairwise $L^2$ difference & $2.35\times10^{-1}$ & $2.11\times10^{-1}$ \\
Max.\ relative mass drift (QLBM) & $6.72\times10^{-4}$ & $1.08\times10^{-4}$ \\
Multiplier admissibility & $|\lambda_r|\le1$ (fixed) & $|\lambda_\nu|<1$, $\nu_t\ge0$, $\tau_{\mathrm{eff}}>1/2$ \\
\bottomrule
\end{tabular}
\end{table}

\begin{table}[t]
\centering
\caption{Post-transient statistics from the production calculations
($t\geq 3\times10^{5}$).
DNS and classical LES run to $10^{6}$ steps; digital LES to $8.78\times10^{5}$.
$\langle K\rangle$ and $\langle\varepsilon_{\mathrm{mol}}\rangle$ are means of the
kinetic-energy and molecular spectral-dissipation histories.
$\sigma_K$ and $\sigma_{\varepsilon}$ are the corresponding sample
standard deviations.
$\langle R_\varepsilon\rangle$ is the mean residual over the same window.
LES uses $\varepsilon=\varepsilon_\nu+\varepsilon_{\mathrm{SGS}}$ in that residual.
Mass drift is $\max_t|m(t)/m(0)-1|$ over the full integration.
Classical values use $\rho$-mean drift at constant volume.}
\label{tab:post-transient}
\resizebox{\linewidth}{!}{%
\begin{tabular}{@{}lcccccc@{}}
\toprule
Case & $\langle K\rangle$ & $\sigma_K$ & $\langle\varepsilon_{\mathrm{mol}}\rangle$ & $\sigma_{\varepsilon}$ & $\langle R_\varepsilon\rangle$ & Mass drift \\
\midrule
Classical MRT DNS & $1.602\times10^{-3}$ & $2.571\times10^{-4}$ & $1.652\times10^{-7}$ & $3.224\times10^{-8}$ & $1.041\times10^{-1}$ & $8.67\times10^{-11}$ \\
Digital QLBM DNS & $1.538\times10^{-3}$ & $2.455\times10^{-4}$ & $1.597\times10^{-7}$ & $2.917\times10^{-8}$ & $1.022\times10^{-1}$ & $6.72\times10^{-4}$ \\
Classical MRT LES & $2.908\times10^{-4}$ & $5.064\times10^{-5}$ & $4.544\times10^{-8}$ & $6.418\times10^{-9}$ & $6.295\times10^{-2}$ & $8.95\times10^{-11}$ \\
Digital QLBM LES & $2.853\times10^{-4}$ & $4.295\times10^{-5}$ & $4.496\times10^{-8}$ & $5.535\times10^{-9}$ & $3.504\times10^{-2}$ & $1.08\times10^{-4}$ \\
\bottomrule
\end{tabular}%
}
\end{table}

One complete collide--force--stream step on an $8^3$ lattice agrees with
the parent floating-point MRT kernels to within fixed-point truncation
(\cref{tab:verification}).
The LES test matches the strain magnitude, eddy viscosity, and the
state-dependent shear multiplier to the same tolerance.

Turbulent dynamics are sensitive to small perturbations
\cite{pope2000turbulent}.
Finite-width perturbations therefore cause the classical and register
trajectories to decorrelate pointwise over long integrations, while the
corresponding exact-real timestep maps still agree
(\cref{eq:commuting}).
Long-horizon agreement is statistical: mass, energy budgets, spectra,
dissipation, intermittency, and the post-transient means in
\cref{tab:post-transient}.
Final population fields show the same pointwise mix
(\cref{fig:population-comparison}).

\subsection{Kinetic energy and dissipation}
The history figures use the four matched production calculations of
\cref{tab:numerical-setup}
(classical MRT and digital QLBM, DNS $320^3$ and LES $64^3$).
Time is lattice time in eddy turnovers
$T_{\mathrm{eddy}}=u_{\mathrm{rms}}^{2}/\varepsilon_{\mathrm{tot}}$,
one value per case.
DNS and LES therefore span different
$t/T_{\mathrm{eddy}}$ ranges at the production step counts of \cref{tab:numerical-setup}.
$K$ is the shell sum of \cref{eq:hit-scales}.
Spectral dissipation in this subsection is the molecular $\varepsilon_{\mathrm{mol}}$ of that equation.
\begin{figure}[H]
 \centering
 \includegraphics[width=0.78\linewidth]{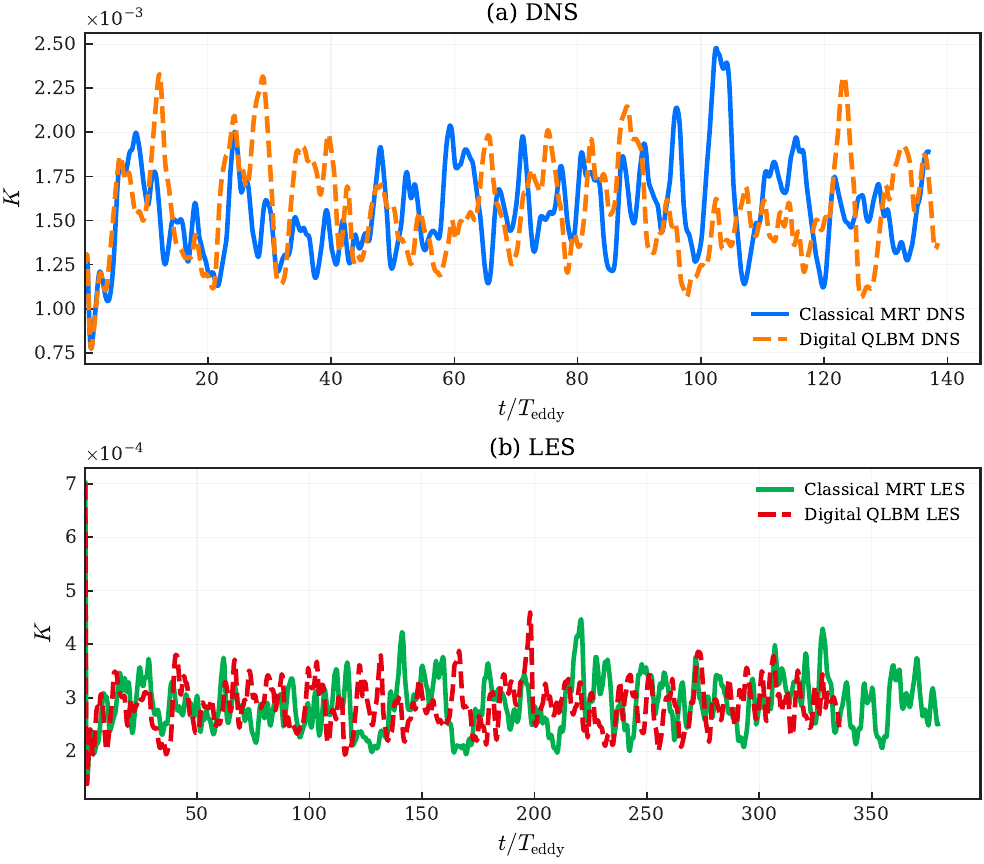}
 \caption{Volume-averaged turbulent kinetic energy $K$.
 (a)~DNS.
 (b)~LES.}
 \label{fig:tke}
\end{figure}

Both pairs occupy the same statistically stationary energy band
(\cref{fig:tke}). Post-transient mean kinetic energies differ by
approximately $4.0\%$ in DNS and $1.9\%$ in LES, with comparable
fluctuation levels (\cref{tab:post-transient}).

\begin{figure}[H]
 \centering
 \includegraphics[width=0.78\linewidth]{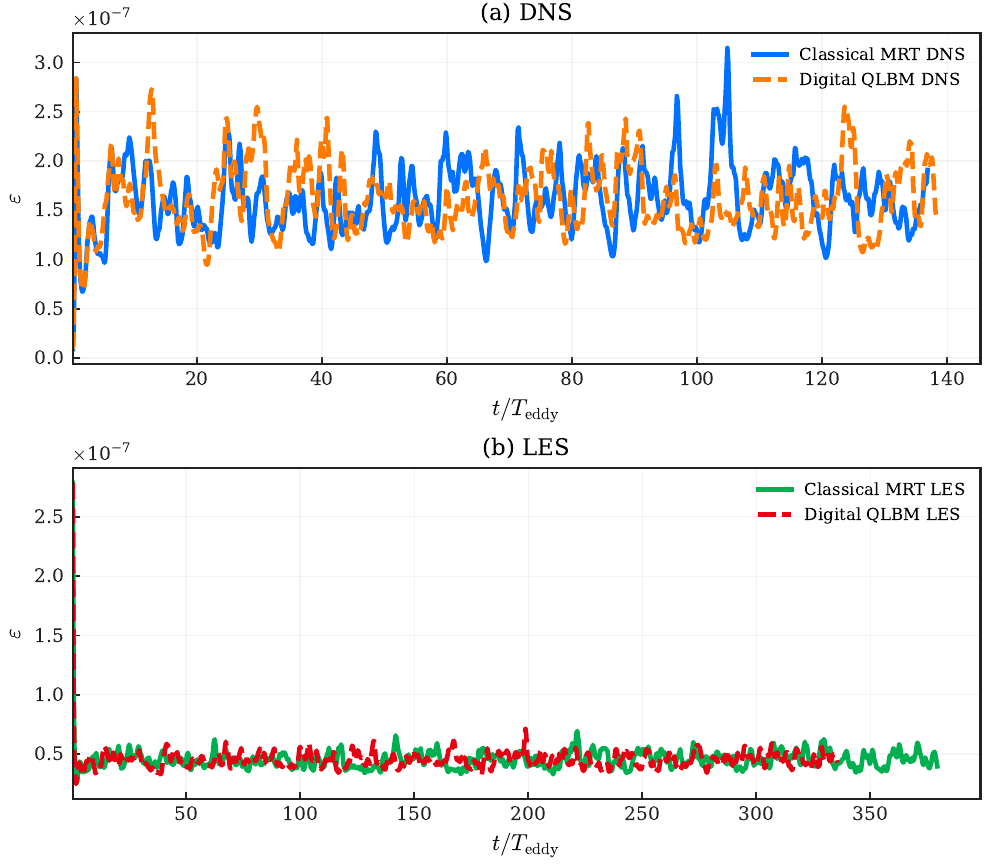}
 \caption{Spectral dissipation $\varepsilon_{\mathrm{mol}}$.
 (a)~DNS.
 (b)~LES (molecular $\nu_0$ only, with $\varepsilon_{\mathrm{SGS}}$ entering
 \cref{fig:energy-balance}).}
 \label{fig:dissipation}
\end{figure}

Post-transient means of this spectral $\varepsilon_{\mathrm{mol}}$ differ by
approximately $3.4\%$ (DNS) and $1.1\%$ (LES), with comparable
sample standard deviations (\cref{fig:dissipation,tab:post-transient}).
The lower LES level is the coarser $64^3$ filtered field.

\begin{figure}[H]
 \centering
 \includegraphics[width=0.78\linewidth]{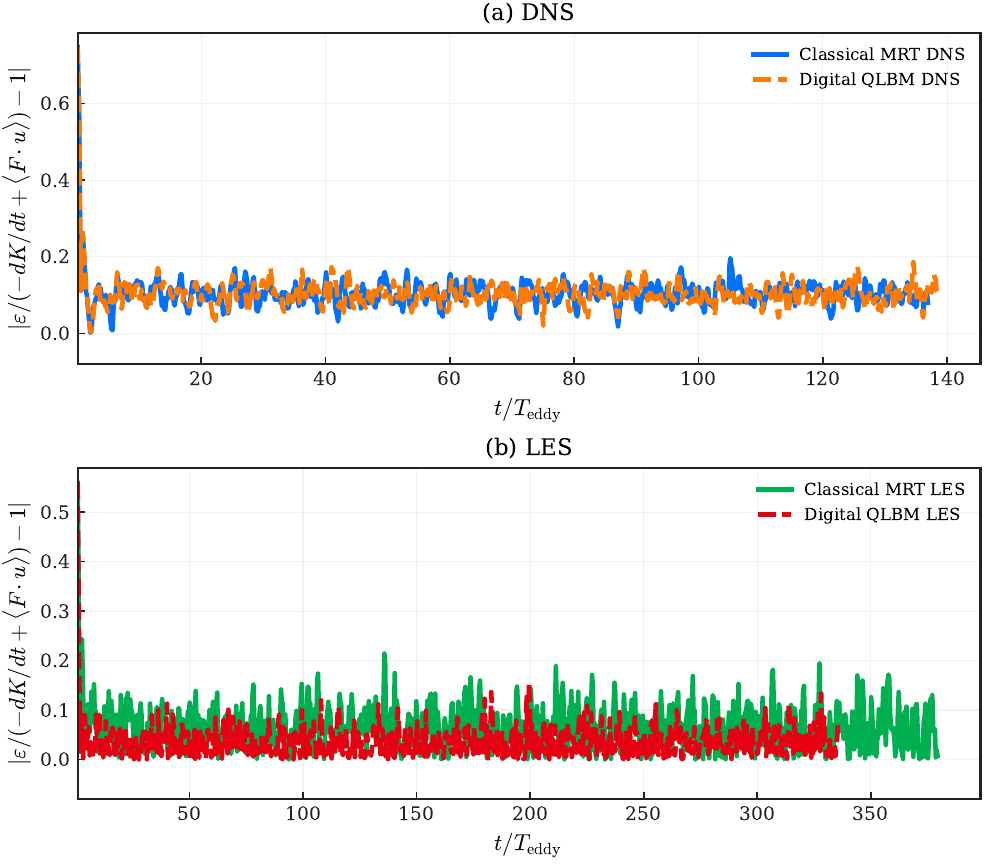}
 \caption{Energy-budget residual
 $R_\varepsilon=
 |\varepsilon_{\mathrm{tot}}/(-dK/dt+P_F)-1|$,
 with $P_F=\langle(\vect{F}/\rho)\cdot\vect{u}\rangle$.
 Strain-based $\varepsilon_\nu=2\nu_0\langle S_{\alpha\beta}S_{\alpha\beta}\rangle$
 differs from the spectral $\varepsilon_{\mathrm{mol}}$ of \cref{fig:dissipation}.
 (a)~DNS ($\varepsilon_{\mathrm{tot}}=\varepsilon_\nu$).
 (b)~LES ($\varepsilon_{\mathrm{tot}}=\varepsilon_\nu+\varepsilon_{\mathrm{SGS}}$).}
 \label{fig:energy-balance}
\end{figure}

Post-transient mean residuals are
$\langle R_\varepsilon\rangle\simeq0.104$ (classical) and $0.102$
(register) for DNS, and $0.0630$ and $0.0350$ for LES
(\cref{fig:energy-balance,tab:post-transient}).
The floor is a per-sample discrete-sampling residual of the forced HIT
budget.
Mean LES imbalances are $3.8\times10^{-9}$ (classical) and
$-1.1\times10^{-9}$ (register), against forcing powers of
$7.1\times10^{-8}$ and $7.0\times10^{-8}$.
The LES energy budgets give
$\langle\varepsilon_{\nu}\rangle=3.50\times10^{-8}$ and
$\langle\varepsilon_{\mathrm{SGS}}\rangle=3.94\times10^{-8}$ (classical)
versus $3.46\times10^{-8}$ and $3.38\times10^{-8}$ (register), hence
$\langle\varepsilon_{\mathrm{tot}}\rangle=7.44\times10^{-8}$ versus
$6.84\times10^{-8}$.
The register LES carries less subgrid dissipation and a comparable
resolved viscous component.
The dissipation column of \cref{tab:post-transient} is the molecular
spectral mean $\langle\varepsilon_{\mathrm{mol}}\rangle$.

\begin{figure}[H]
 \centering
 \includegraphics[width=0.49\linewidth]{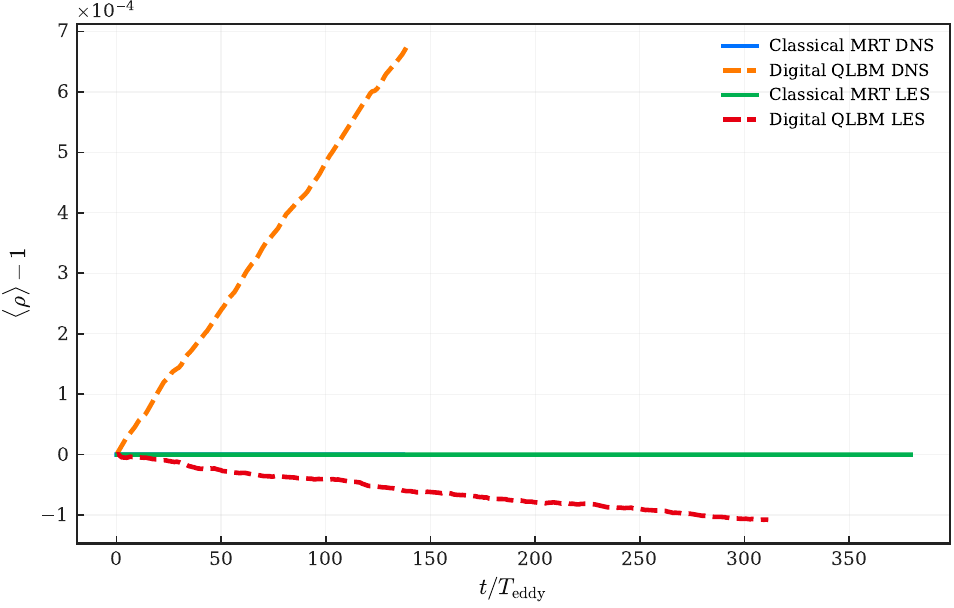}
 \includegraphics[width=0.49\linewidth]{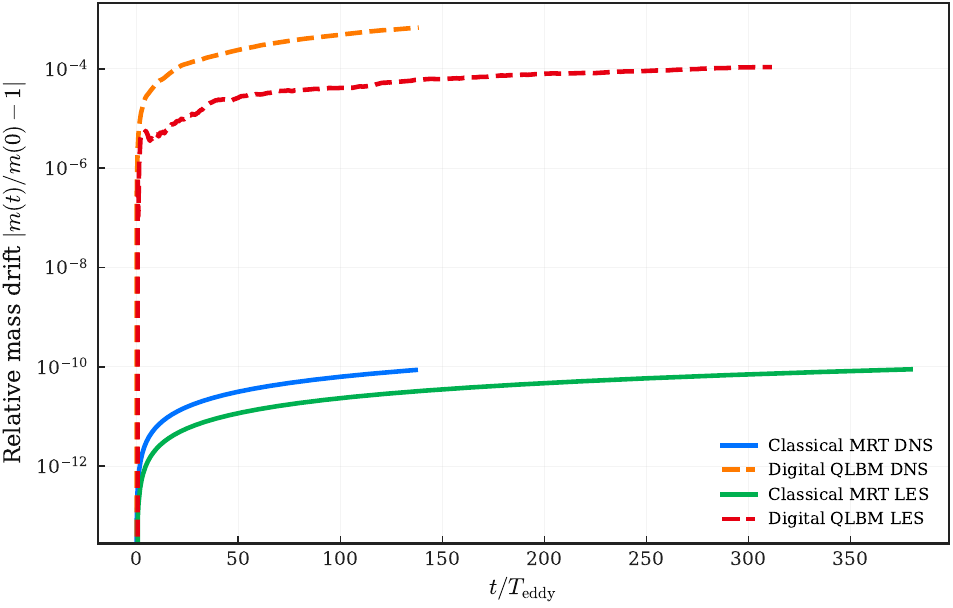}
 \caption{Mass conservation.
 Left, signed density deviation $\langle\rho\rangle-1$.
 Right, relative mass drift $|m(t)/m(0)-1|$ (log scale).
 Classical float64 mass drift remains $\sim10^{-10}$.
 Digital QLBM at production width $b=24$ reaches $6.72\times10^{-4}$ (DNS) and
 $1.08\times10^{-4}$ (LES).}
 \label{fig:mass}
\end{figure}

Maximum relative mass drifts are
$8.67\times10^{-11}$ (DNS) and $8.95\times10^{-11}$ (LES) for the
classical float64 solvers, versus $6.72\times10^{-4}$ (DNS) and
$1.08\times10^{-4}$ (LES) for the digital fixed-point QLBM-register
calculations (\cref{fig:mass,tab:post-transient}).
The extra register drift is finite-bit-width rounding of $\Phi_b$ at
$b=24$.

\subsection{QLBM-DNS}
\begin{figure}[H]
 \centering
 \includegraphics[width=0.96\linewidth]{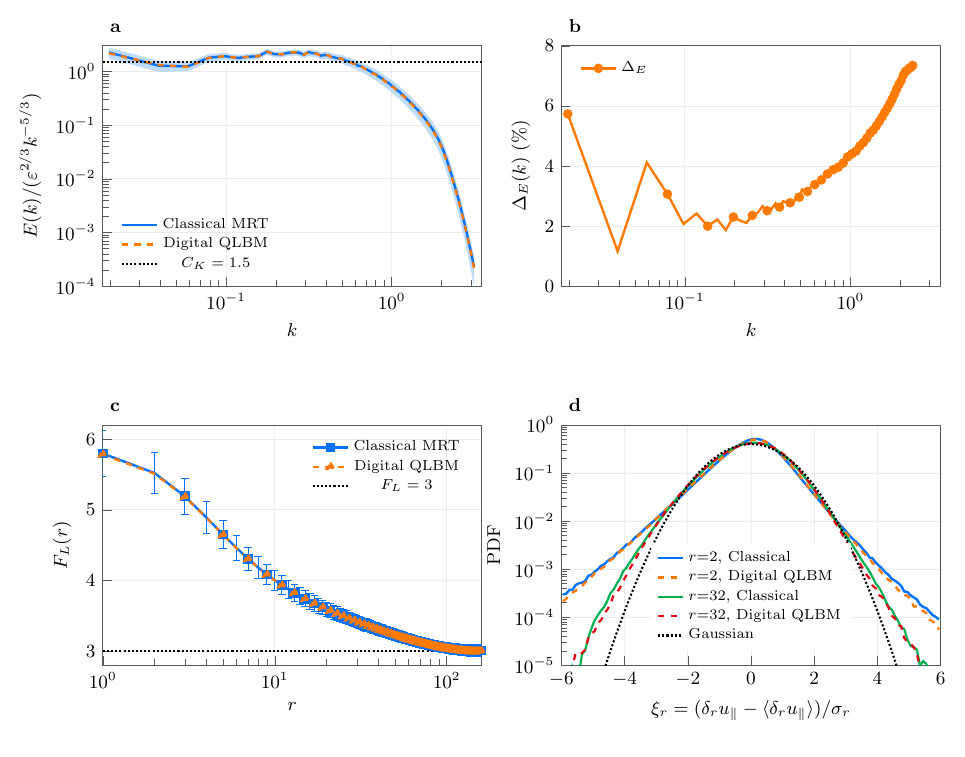}
 \caption{Multiscale DNS fidelity of classical MRT and digital QLBM-register
 solvers on $320^3$ (production fields).
 a~Time-averaged Kolmogorov-compensated spectra
 $E(k)/(\varepsilon_{\mathrm{mol}}^{2/3}k^{-5/3})$ with sample standard-deviation band
 and $C_K=1.5$ reference.
 b~Scale-resolved spectral discrepancy
 $\Delta_E(k)=100\,|\overline E_{\mathrm{reg}}-\overline E_{\mathrm{cl}}|
 /\overline E_{\mathrm{avg}}$, where
 $\overline E_{\mathrm{avg}}=(\overline E_{\mathrm{reg}}+\overline E_{\mathrm{cl}})/2$.
 Bins with
 $\overline E_{\mathrm{avg}}/\max_k\overline E_{\mathrm{avg}}\leq10^{-6}$
 are omitted where the spectral energy is negligibly small.
 c~Longitudinal increment flatness $F_L(r)$ (\cref{eq:flatness}).
 Error bars show sample standard deviation. Dotted line, Gaussian $F_L=3$.
 d~Probability densities of the signed increments
 $\xi_r=(\delta_r u_{\parallel}-\langle\delta_r u_{\parallel}\rangle)/\sigma_r$
 at $r/\Delta x=2$ ($r/\eta\simeq4$) and $r/\Delta x=32$ ($r/\eta\simeq65$).
 Both solvers use the same post-transient sampling times, histogram edges, and
 Cartesian average, with spatial stride~$2$.
 Each distribution is normalized independently to zero mean and unit variance.
 Histogram bins are retained, including low counts.}
 \label{fig:dns-multiscale}
\end{figure}

Classical and register DNS spectra overlap across the resolved band
(time-averaged relative difference $5.45\times10^{-2}$,
\cref{tab:verification},
\cref{fig:dns-multiscale}\,a,b).
A short near-inertial interval sits near the $C_K=1.5$ reference before
the viscous roll-off.
At $\mathrm{Re}_\lambda\simeq220$ that interval is short.
Flatness and increment PDFs agree within the observed sampling
variability (\cref{fig:dns-multiscale}\,c,d).
$F_L(r)$ is elevated at small $r$ (intermittency) and approaches the
Gaussian limit at large separations.
ESS scaling for the same DNS pair is shown with the LES cases in
\cref{fig:ess}.

\subsection{QLBM-LES}
\begin{figure}[H]
 \centering
 \includegraphics[width=0.78\linewidth]{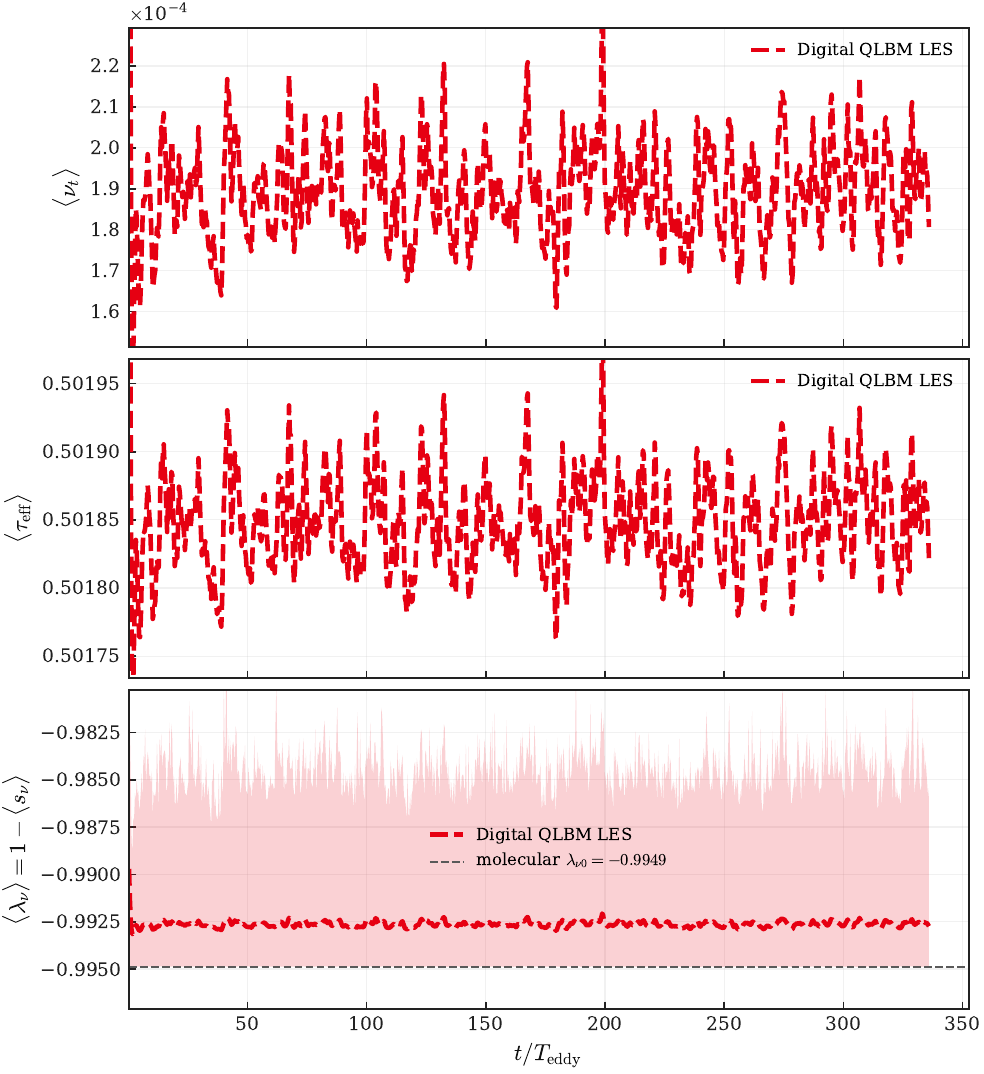}
 \caption{Digital QLBM-register LES state-dependent Smagorinsky control path on
 $64^3$.
 Top, domain-mean eddy viscosity $\langle\nu_t\rangle$.
 Middle, effective shear relaxation time $\langle\tau_{\mathrm{eff}}\rangle$.
 Bottom, shear multiplier $\langle\lambda_\nu\rangle=1-\langle s_\nu\rangle$,
 zoomed about the molecular value $\lambda_{\nu0}=1-1/\tau_0$ (dashed).
 The shaded band is the instantaneous spatial min--max of $\lambda_\nu$.
 Because $\tau_0\simeq1/2$, $\lambda_\nu$ remains near $-1$ and inside
 $(-1,1)$, with overrelaxation through negative $L$.}
 \label{fig:les-multiplier}
\end{figure}

The classical LES uses the same discrete Smagorinsky rule.
\Cref{fig:les-multiplier} shows the register control path is active, with
$\lambda_\nu$ inside $(-1,1)$.

\begin{figure}[H]
 \centering
 \includegraphics[width=0.72\linewidth]{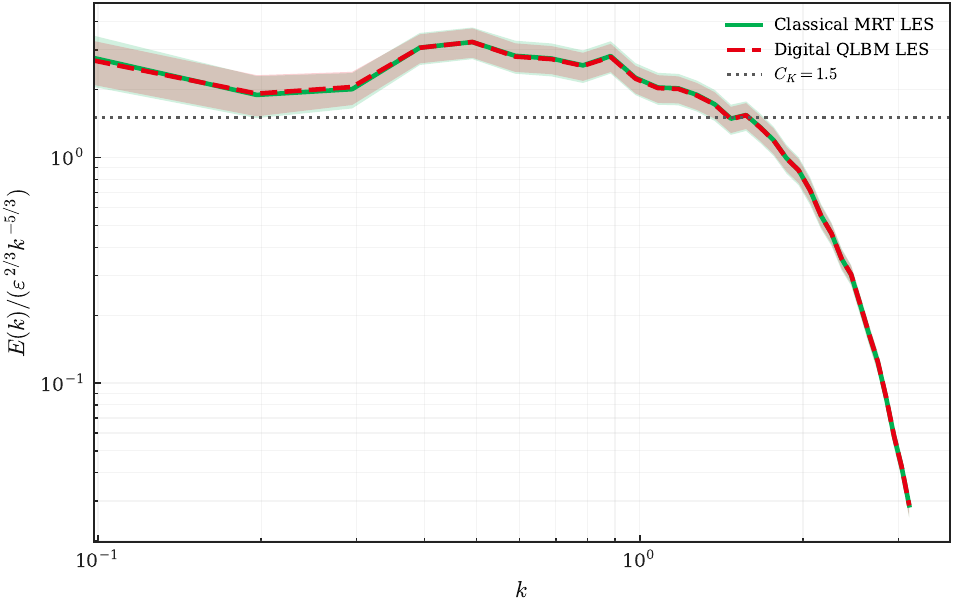}
 \caption{Time-averaged compensated spectra for the matched $64^3$
 Smagorinsky LES calculations (classical MRT and digital QLBM-register),
 using the same compensation as \cref{fig:dns-multiscale}\,a.
 Shaded band, sample standard deviation of $E(k)$.
 Dotted line, $C_K=1.5$.}
 \label{fig:spectra-les}
\end{figure}

Classical and register LES spectra overlap across the resolved band
(time-averaged relative difference $3.05\times10^{-2}$,
\cref{tab:verification,fig:spectra-les}).
Relative to the DNS compensated spectra, the LES curves sit somewhat
higher and are less smooth.
On the coarse $64^3$ Smagorinsky grid the near-inertial interval is short.
The $C_K=1.5$ line is a visual reference.
Classical and register DNS and LES spectra are further compared with the
Pope model (\cref{eq:pope}) in \cref{fig:pope}.

\begin{figure}[H]
 \centering
 \includegraphics[width=0.72\linewidth]{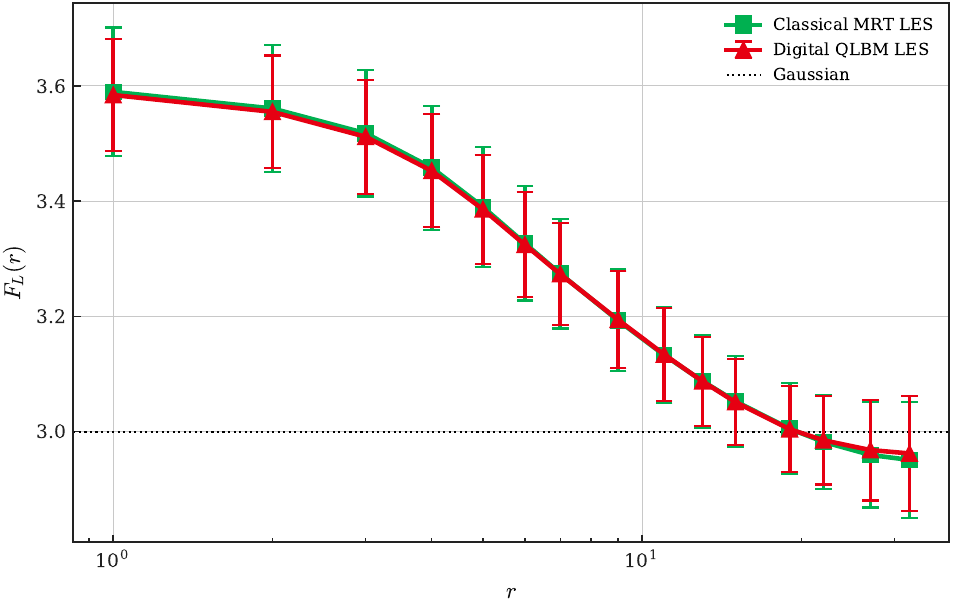}
 \caption{LES longitudinal increment flatness $F_L(r)$
 (\cref{eq:flatness}) for classical MRT and digital QLBM-register solvers on
 $64^3$.
 Error bars show sample standard deviation across snapshots. Dotted line,
 Gaussian $F_L=3$.}
 \label{fig:flatness-les}
\end{figure}

Classical and register curves coincide
(\cref{fig:flatness-les}).
Relative to DNS (\cref{fig:dns-multiscale}\,c), the small-scale elevation is
weaker ($F_L(1)\simeq3.6$ versus $\simeq5.8$) and the Gaussian limit is
reached at smaller separations, as expected for filtered Smagorinsky
LES.
In ESS (\cref{fig:ess}), classical and register $S_p$--$S_3$ curves
overlap within each regime.
Low-order exponents ($p=1,2,3$) stay close between DNS and LES.
At $p=4$--$6$ the LES anomalies are weaker
($\xi_6-2\simeq-0.11$ versus $\simeq-0.24$ in DNS),
again as expected for the filtered field.

\begin{figure}[H]
 \centering
 \includegraphics[width=0.88\linewidth]{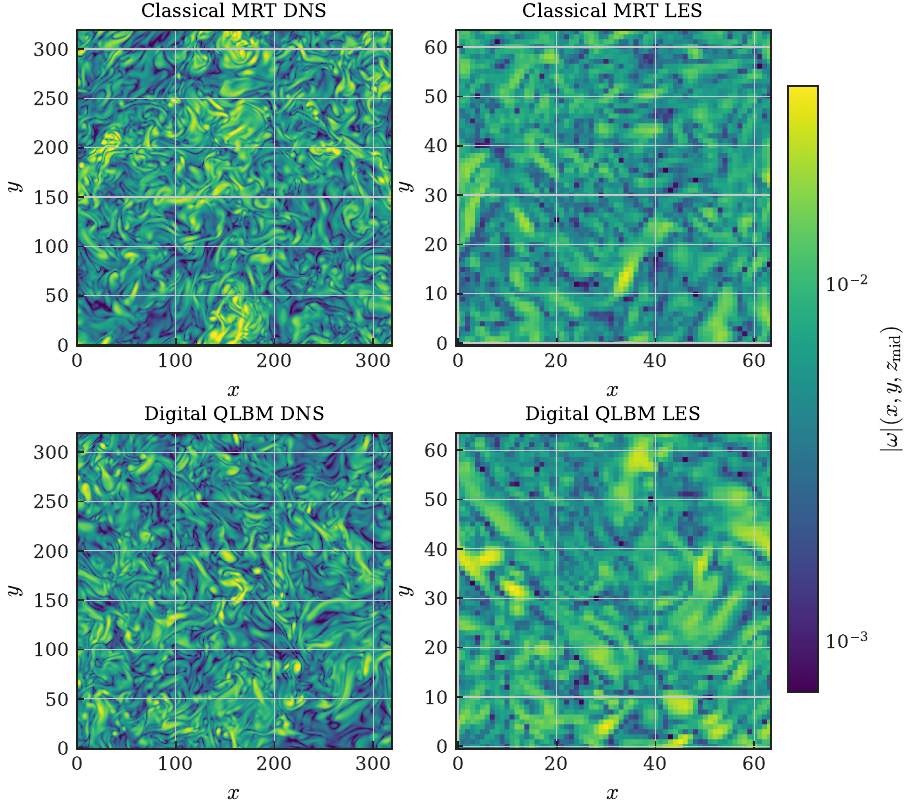}
 \caption{Final-state midplane vorticity magnitude
 $|\omega|(x,y,z_{\mathrm{mid}})$, with a shared logarithmic colour scale.
 Left column, $320^3$ DNS after $10^{6}$ lattice steps. Right column, $64^3$
 Smagorinsky LES (classical after $10^{6}$ steps; digital after $8.78\times10^{5}$).
 Top row, classical MRT. Bottom row, digital QLBM-register.}
 \label{fig:vorticity-midplane}
\end{figure}

The DNS panels of \cref{fig:vorticity-midplane} resolve fine vortical filaments, whereas the LES panels
show the coarser filtered field on the $64^3$ grid.
Within each regime the classical and register snapshots show similar vortical structure.

Every local shear multiplier stays inside $(-1,1)$
(\cref{fig:les-multiplier}).
Residual differences follow from the shared discrete Smagorinsky closure
at finite width.

\subsection{Fixed-point ablation}
\label{sec:fixedpoint-results}
Production turbulence uses the digital fixed-point backend of
Algorithm~\ref{alg:digital-timestep} at fractional width $b=24$.
Forced $32^3$ DNS at $b\in\{16,24,32,40\}$ for $2\times10^{5}$ steps is
compared with a matched floating-point reference
(\cref{app:fixedpoint,tab:fixedpoint-robustness,fig:fixedpoint-robustness}).
Mass drift at $b=24$ is $1.36\times10^{-3}$, falling to $1.87\times10^{-6}$
at $b=32$ and $7.26\times10^{-9}$ at $b=40$
(\cref{tab:fixedpoint-robustness,fig:fixedpoint-robustness}).
At $b=16$ the flow collapses, with final $K$ about three orders below the
reference.
Widths $b=24$, $32$, and $40$ remain in the sustained kinetic-energy band.
Pre-chaotic kinetic-energy error is $9.24\times10^{-3}$ at $b=24$,
$1.45\times10^{-5}$ at $b=32$, and $6.95\times10^{-8}$ at $b=40$.
Over the full horizon, pairwise field $L^2$ for those three widths sits
near $2\times10^{-2}$ and is mixed by chaos.
Production HIT uses $b=24$ on the $320^3$/$64^3$ grids of
\cref{tab:numerical-setup}.
The $32^3$ survey is only the reduced-grid check of that bit width
(\cref{tab:fixedpoint-robustness}).

\section{Circuit-uplift requirements and resource estimates}
The register intermediate representation defines a circuit-level replacement path for every numerical kernel under Algorithm~\ref{alg:digital-timestep}.
Fixed sparse transforms $M$ and $M^{-1}$ compile to reversible signed adders and constant multipliers
\cite{vedral1996arithmetic,haner2018optimizing}.
Density, momentum, equilibrium, forcing, and LES-control quantities use reversible fixed-point arithmetic.
Reciprocal and square-root functions use reversible Newton iteration on the fixed-point grid
\cite{haner2018optimizing}, with each square-root step calling the reciprocal.
Periodic streaming is a fixed reversible permutation of population-word locations.
MRT relaxation is a reversible signed multiply.
Open-system irreversibility is unconditional $\mathrm{Reset}$ of bank $A$
(\cref{prop:cptp-reset}).
\Cref{fig:qlbm-register-lifetime} shows the logical register lifetime of one forced DNS timestep.
The circuit evaluates \(U_{\Phi}\) into bank \(B\), uncomputes the arithmetic workspace, and applies \(\mathrm{Reset}(A)\).

\begin{sidewaysfigure}[p]
\centering
\includegraphics[width=0.88\textheight,height=0.90\textwidth,keepaspectratio]{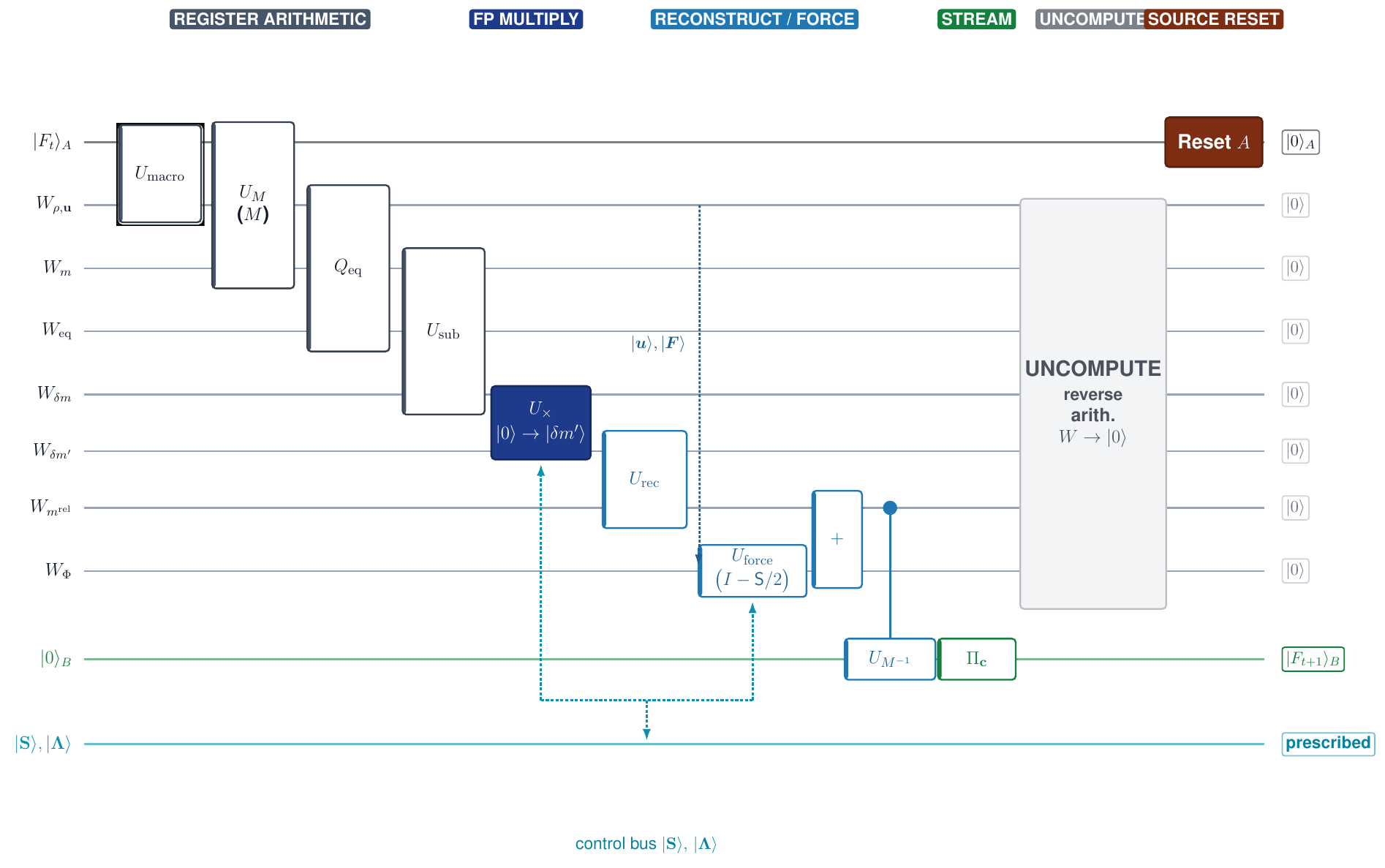}
\caption{Logical register lifetime of one forced D3Q19 DNS timestep
(companion to \cref{fig:qlbm-register-dns-les}).
Horizontal lines are register families.
The signed multiply $U_{\times}$ writes $\delta m'=\bm\Lambda\odot\delta m$ into a clean word $W_{\delta m'}$.
$U_{\mathrm{force}}$ is the discrete body-force source.
$U_{\Phi}$ in \cref{eq:u-step} is the full two-bank map.
MRT rates $\ket{\bm S},\ket{\bm\Lambda}$ are prescribed.
Streaming writes bank $B$.
Uncomputation returns the workspace to $\ket{0}$, after which $\mathrm{Reset}(A)$ discards the obsolete source.
The same arithmetic primitives are compiled at gate level on D2Q9 (\cref{app:d2q9}).}
\label{fig:qlbm-register-lifetime}
\end{sidewaysfigure}

Each workspace line is a register family, restored to $\ket{0}$ after streaming into $B$ and therefore reusable on the next site.
DNS uses prescribed MRT rates.
The LES construction of local multipliers from $F_t$ is given in \cref{fig:qlbm-register-dns-les}\,b.
Persistent population storage is one bank
\begin{equation}
 W_{\mathrm{pop}}=Q N^{d}\,b_{\mathrm{pop}},
 \label{eq:pop-width}
\end{equation}
and two banks in the architecture, $W_{\mathrm{state}}=2W_{\mathrm{pop}}$.
Temporary arithmetic workspace is reusable under serial site scheduling, with a schedule-dependent peak width of order $10^{3}$--$10^{4}$ logical qubits.
Under that schedule the full-lattice Toffoli count and serial depth per
timestep scale as
$G_{\mathrm{step}}=N^{d}G_{\mathrm{site}}$ and
$D_{\mathrm{serial}}=N^{d}D_{\mathrm{site}}$.
For the production D3Q19 grids, $Q=19$ and $d=3$.
With the analytic per-site band $G_{\mathrm{site}}=10^{4}$--$10^{5}$, $N=320$ is of order $10^{11}$--$10^{12}$ Toffoli-equivalent operations per timestep before the $10^{6}$-step horizon.
\Cref{tab:resource-estimate} reports the corresponding bank widths, reset counts, and analytic leading-order arithmetic resources at $b_{\mathrm{pop}}=1{+}2{+}24=27$, using the production Newton counts $J_{\mathrm{r}}=J_{\mathrm{s}}=6$.
The $T$-count uses seven $T$ gates per Toffoli-equivalent, so the per-site band is $7\times10^{4}$--$7\times10^{5}$.
The logical circuit measures at the final readout.
A hardware reset may compile as measure-and-correct.
Two-bank storage and serial full-lattice depth limit feasibility at production scale.
Gate-level confirmation of the same compute--reset contract on D2Q9 is given in \cref{app:d2q9}.

\begin{table}[!htbp]
\centering
\caption{Logical-resource estimate for one forced D3Q19 digital-register timestep
under the two-bank compute--reset architecture of \cref{eq:boxed-algo,prop:cptp-reset}.
Persistent bank widths and reset counts follow from $W_{\mathrm{pop}}=Q N^{d}b_{\mathrm{pop}}$
with $Q=19$, $d=3$, and the production signed population word $b_{\mathrm{pop}}=1{+}2{+}24=27$;
the second population bank has equal width, so $W_{\mathrm{state}}=2W_{\mathrm{pop}}$.
Temporary workspace, Toffoli-equivalent counts, $T$-counts, and depths are analytic
leading-order D3Q19 estimates from reversible fixed-point arithmetic and the register-level IR.
Compiled D2Q9 gate-level counts appear in \cref{app:d2q9}.
Work-register qubits are reusable across sites in a serial schedule.
Dissipation is unconditional reset of the obsolete bank, with algorithmic measurement at the final readout.
Full-lattice serial cost scales as $N^{d}$ times the per-site count.}
\label{tab:resource-estimate}
\footnotesize
\setlength{\tabcolsep}{5pt}
\begin{tabular}{@{}>{\raggedright\arraybackslash}p{0.40\textwidth}cccc@{}}
\toprule
Resource & $N=8$ & $N=32$ & $N=64$ & $N=320$ \\
 & one-step & width survey & LES & DNS \\
\midrule
Persistent pop.\ qubits $W_{\mathrm{pop}}$ (bank $A$)
  & $2.6\times10^{5}$ & $1.7\times10^{7}$ & $1.3\times10^{8}$ & $1.7\times10^{10}$ \\
Second population bank $B$
  & \multicolumn{4}{c}{$W_{\mathrm{pop}}$} \\
Persistent state $W_{\mathrm{state}}=2W_{\mathrm{pop}}$
  & $5.3\times10^{5}$ & $3.4\times10^{7}$ & $2.7\times10^{8}$ & $3.4\times10^{10}$ \\
Temp.\ arithmetic workspace (serial)
  & \multicolumn{4}{c}{$O(10^{3}$--$10^{4})$} \\
Integer / fractional bits, excl.\ sign (pop.\ / moment)
  & \multicolumn{4}{c}{$2{+}24$ / $6{+}24$} \\
Signed word widths (pop.\ / moment)
  & \multicolumn{4}{c}{$27$ / $31$} \\
Adder-tree / multiply ops (order, per site)
  & \multicolumn{4}{c}{$O(10^{2})$} \\
Newton reciprocal / sqrt ($J_{\mathrm{r}}=J_{\mathrm{s}}=6$)
  & \multicolumn{4}{c}{$O(J_{\mathrm{r}}w^{2})$ / $O(J_{\mathrm{s}}J_{\mathrm{r}}w^{2})$} \\
Toffoli-eq.\ / site / step (analytic)
  & \multicolumn{4}{c}{$10^{4}$--$10^{5}$} \\
Serial Toffoli / step (all sites)
  & $5\times10^{6}$--$5\times10^{7}$
  & $3\times10^{8}$--$3\times10^{9}$
  & $3\times10^{9}$--$3\times10^{10}$
  & $3\times10^{11}$--$3\times10^{12}$ \\
Estimated $T$-count / site / step (assuming $7T$ per Toffoli-eq.)
  & \multicolumn{4}{c}{$7\times10^{4}$--$7\times10^{5}$} \\
Serial-depth upper estimate / site / step
  & \multicolumn{4}{c}{$10^{4}$--$10^{5}$} \\
Serial logical depth / step (all sites)
  & $5\times10^{6}$--$5\times10^{7}$
  & $3\times10^{8}$--$3\times10^{9}$
  & $3\times10^{9}$--$3\times10^{10}$
  & $3\times10^{11}$--$3\times10^{12}$ \\
Qubit resets / site / step (obsolete bank)
  & \multicolumn{4}{c}{$19b_{\mathrm{pop}}=513$} \\
Algorithmic intermediate measurements / step
  & \multicolumn{4}{c}{$0$} \\
Final readout measurements
  & \multicolumn{4}{c}{$W_{\mathrm{pop}}$} \\
\bottomrule
\end{tabular}

\vspace{0.35em}
{\footnotesize
A full-bank reset is $W_{\mathrm{pop}}$ qubits.
The square-root cost is the adopted arithmetic model: each of the $J_{\mathrm{s}}$ outer Newton steps uses a $J_{\mathrm{r}}$-iteration reciprocal.
LES local-control arithmetic sits toward the upper end of the per-site Toffoli band.
Persistent two-bank width and serial $N^{d}$ depth both constrain feasibility.
Logical reset may compile as hardware-specific measurement and feedforward.
The algorithm measures at the final readout.
}
\end{table}

\clearpage
\section{Discussion and limitations}
Linear and advection--diffusion QLBMs use amplitude or linearized encodings.
Wawrzyniak et al.\ and Tiwari et al.\ treat ADE with measurement and reinitialization requirements that depend on the update
\cite{wawrzyniak2025qlbm,tiwari2025realizable}.
Nagel and L{\"o}we iterate linear ADE for multiple steps without mid-run extraction or reinitialization
\cite{nagel2026multistep}.
Carleman embeddings lift the nonlinear collision into a larger linear space
\cite{sanavio2024carleman,bastida2026carleman,khan2026lindblad}.
Ensemble and denoising methods combine a different state representation with sampling or intermediate reconstruction
\cite{wang2025nonlinear,duong2026denoising}.
The present architecture evaluates a deterministic digital timestep by reversible arithmetic into a clean bank and an unconditional reset of the obsolete source.

Signed-rail amplitude damping is a related CPTP LBM primitive, with a different encoding and a different dissipation mechanism
\cite{khan2026deterministic}.

In exact real arithmetic, \cref{eq:commuting} identifies the ideal D3Q19
map with the matched classical forced MRT map, including LES
when both sides evaluate the same discrete $G$.
The million-step calculations support that identity statistically through
mass, energy budgets, spectra, dissipation, and intermittency
(\cref{tab:post-transient}).
Chaotic trajectories still decorrelate under truncation-scale perturbations
\cite{pope2000turbulent}.

The D3Q19 calculations use the register emulator of the circuit.
D2Q9 is the compiled gate-level check on a reduced lattice
(\cref{app:d2q9}).
Persistent two-bank width $W_{\mathrm{state}}=2Q N^{d}b_{\mathrm{pop}}$
is prohibitive at production scale, and serial full-lattice depth is a
further constraint
(\cref{tab:resource-estimate}).
On $32^3$, $b=16$ collapses, while $b=24$, $32$, and $40$ remain in the
sustained kinetic-energy band
(\cref{sec:fixedpoint-results,app:fixedpoint}).
Production HIT uses $b=24$ on the large grids.

\section{Conclusion}
We constructed a complete digital-register QLBM for the conventional
nonlinear MRT timestep.
Macroscopic recovery, nonlinear equilibrium, signed MRT collision, body
forcing, streaming, and state-dependent Smagorinsky LES control are
evaluated as fixed-point register operations.
Reversible evaluation into an alternating destination bank, followed by
unconditional source-bank reset, is a CPTP channel that composes without
intermediate decoding or classical state reinjection
(\cref{prop:cptp-reset,eq:boxed-algo}).
The D3Q19 calculations treat forced HIT in DNS and LES for $10^6$ timesteps
(Algorithm~\ref{alg:digital-timestep}) on the register emulator of the
circuit.
A separate D2Q9 implementation compiles the same compute--reset contract to
reversible gates under final measurement
(\cref{app:d2q9}).
Million-step matched calculations at Reynolds number 15000 show statistical
agreement across mass, energy budgets, spectra, and intermittency.
Production uses fractional width $b=24$.
Practical feasibility is controlled by two-bank persistent width
$W_{\mathrm{state}}=2Q N^{d}b_{\mathrm{pop}}$, serial full-lattice depth, and
the logical-qubit and reset costs of \cref{tab:resource-estimate}.

\section*{Data and code availability}
Data and code are available from the corresponding author upon reasonable request.

\appendix
\renewcommand{\thesection}{\Alph{section}}
\renewcommand{\theHsection}{appendix.\arabic{section}}
\renewcommand{\theHsubsection}{appendix.\arabic{section}.\arabic{subsection}}
\renewcommand{\theHsubsubsection}{appendix.\arabic{section}.\arabic{subsection}.\arabic{subsubsection}}
\renewcommand{\thesubsection}{\arabic{subsection}}
\counterwithin{equation}{section}
\renewcommand{\theequation}{\thesection\arabic{equation}}
\crefalias{section}{appendix}
\titleformat{\section}[block]
  {\centering\normalfont\normalsize\bfseries}
  {APPENDIX \thesection:}{0.6em}{\MakeTextUppercase}
\titleformat{\subsection}[hang]
  {\normalfont\normalsize\itshape}
  {\thesection\thesubsection.}{0.5em}{}
\section{Supplementary figures}
\label{app:si-figures}
\begin{figure}[H]
 \centering
 \includegraphics[width=\linewidth,height=0.62\textheight,keepaspectratio]{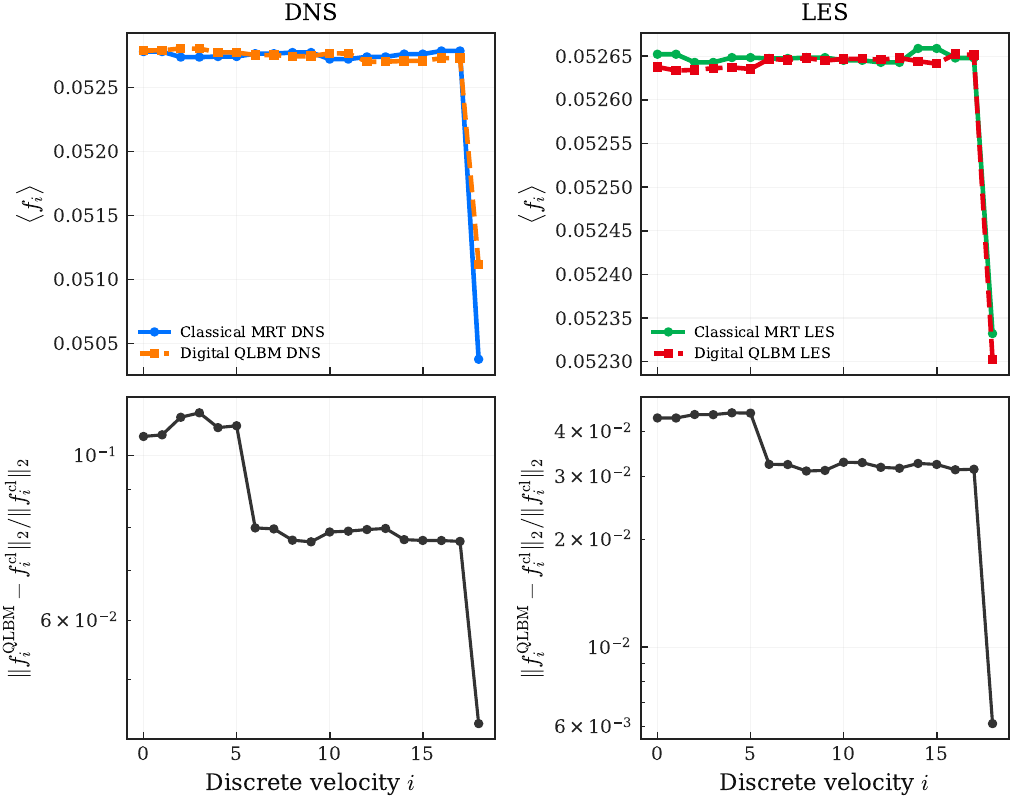}
 \caption{Final-state D3Q19 populations after $10^{6}$ lattice steps.
 Top, domain-mean $\langle f_i\rangle$ for classical MRT and digital QLBM-register
 solvers (left, DNS, right, LES).
 Vertical scales are independently zoomed for each regime.
 Bottom, per-component relative field difference
 $\|f_i^{\mathrm{QLBM}}-f_i^{\mathrm{cl}}\|_2/\|f_i^{\mathrm{cl}}\|_2$.
 The mean profiles nearly overlap, while the final spatial fields differ from a few $10^{-3}$ to $\mathcal{O}(10^{-1})$,
 consistent with long-time chaotic trajectory decorrelation.
 These final-state differences are distinct from the one-step digital-map error in \cref{tab:verification}.}
 \label{fig:population-comparison}
\end{figure}

\begin{figure}[H]
 \centering
 \includegraphics[width=0.72\linewidth]{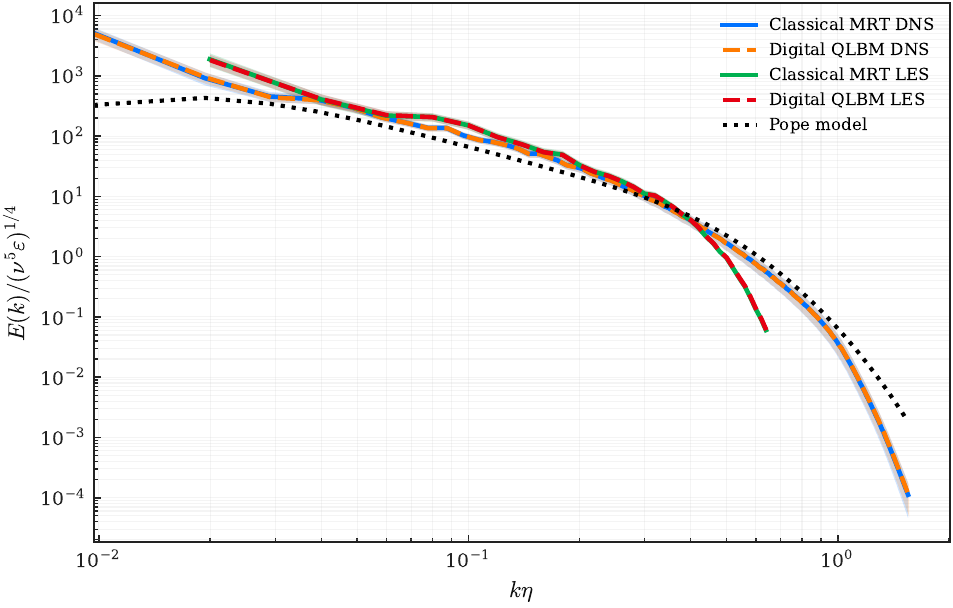}
 \caption{Pope-normalized spectra
 $E(k)/(\nu_0^{5}\varepsilon_{\mathrm{mol}})^{1/4}$ versus $k\eta$ for classical MRT and
 digital QLBM-register DNS and LES, compared with \cref{eq:pope}
 ($C_K=1.5$, $c_L=6.78$, $c_\eta=0.40$, $\beta=5.2$).
 The model curve uses the classical DNS scales $(L_{11},\eta,\varepsilon_{\mathrm{mol}})$
 in the Pope-form spectrum of \cref{eq:pope}.
 Shaded bands, sample standard deviation of $E(k)$.}
 \label{fig:pope}
\end{figure}

\begin{figure}[H]
 \centering
 \includegraphics[width=0.72\linewidth]{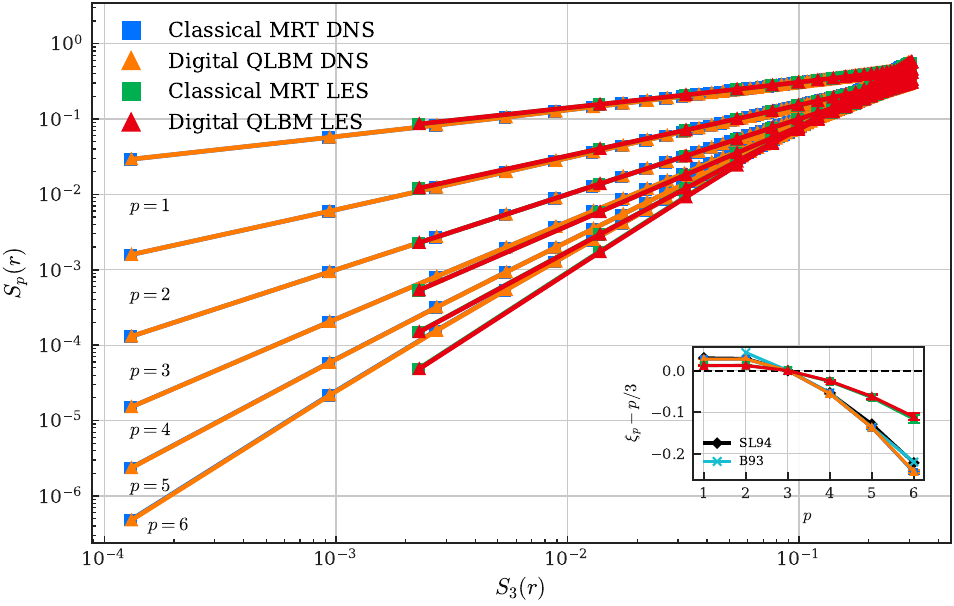}
 \caption{Extended self-similarity for classical MRT and digital QLBM-register
 DNS and LES, $S_p(r)$ versus $S_3(r)$ for $p=1,\ldots,6$
 (\cref{eq:ess}).
 Classical and register curves overlap within each regime.
 Inset, scaling anomalies $\xi_p-p/3$ versus She--Leveque
 \cite{she1994universal} and Benzi \cite{benzi1993ess} references.
 Low-order DNS and LES exponents remain close; at $p=4$--$6$ the LES
 anomalies are weaker than DNS, as expected for the filtered field.}
 \label{fig:ess}
\end{figure}

\section{Fixed-point register-width survey}
\label{app:fixedpoint}
We evaluate the digital map $\Phi_b$ on all register banks of the
DNS emulator and integrate a forced $32^3$ lattice for $2\times10^{5}$
steps at fractional widths $b=16$, $24$, $32$, and $40$, alongside a
matched floating-point reference on the same grid and forcing.
\Cref{tab:fixedpoint-robustness}
summarizes the reduced-grid survey
(and \cref{fig:fixedpoint-robustness}).
Mass drift at $b=40$ is $7.26\times10^{-9}$, against a floating-point
reference of $1.79\times10^{-11}$.
At $b=16$ the flow collapses, with final $K$ of $1.8\times10^{-7}$ against a
reference $9.8\times10^{-5}$.
Widths $b=24$, $32$, and $40$ remain in the sustained kinetic-energy band, with pre-chaotic
kinetic-energy errors $9.24\times10^{-3}$, $1.45\times10^{-5}$, and
$6.95\times10^{-8}$.
Production HIT uses $b=24$.
Final population $L^2$ differences remain near $2\times10^{-2}$ because
that column is a field-shape diagnostic after chaos has mixed the states.

\begin{figure}[H]
 \centering
 \includegraphics[width=\linewidth]{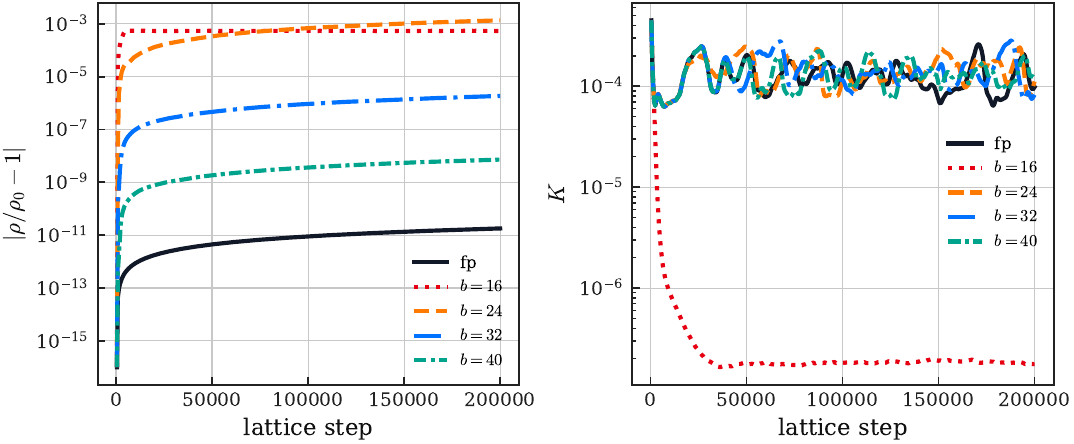}
 \caption{Fixed-point fractional-bit survey on $32^3$ forced DNS
 ($2\times10^{5}$ steps) versus floating point (fp).
 a~Relative mass drift $|\rho/\rho_0-1|$.
 b~Kinetic energy $K$.
 At $b=16$ the flow collapses.
 Widths $b=24$, $32$, and $40$ remain in the sustained kinetic-energy band.
 Pre-chaotic $K$ tracks the floating-point reference more closely as $b$
 increases.
 Production HIT uses $b=24$.}
 \label{fig:fixedpoint-robustness}
\end{figure}

\begin{table}[H]
\centering
\caption{Fixed-point fractional-bit survey on a $32^3$ forced DNS register calculation ($2\times10^{5}$ lattice steps).
Columns~3 and~4 are relative mean kinetic-energy errors $|\langle K\rangle_b-\langle K\rangle_{\mathrm{fp}}|/\langle K\rangle_{\mathrm{fp}}$ over the first $2\times10^{4}$ steps (pre-chaotic window) and the full horizon against the floating-point backend.
Column~5 is the final population relative $L^2$ difference $\|f_b-f_{\mathrm{fp}}\|_2/\|f_{\mathrm{fp}}\|_2$.
Long-horizon pairwise fields mix under chaos, so column~5 is a field-shape diagnostic.
Production HIT uses $b=24$.
This table is the reduced-grid width survey of that format.}
\label{tab:fixedpoint-robustness}
\small
\begin{tabular}{@{}ccccc@{}}
\toprule
 & & \multicolumn{2}{c}{$|\langle K\rangle_b-\langle K\rangle_{\mathrm{fp}}|/\langle K\rangle_{\mathrm{fp}}$} & \\
\cmidrule(lr){3-4}
$b$ & $\max|\Delta\rho/\rho_0|$ & $2\times10^{4}$ & full & $\|f\|_{\mathrm{rel}}$ \\
\midrule
16 & $5.55\times10^{-4}$ & $7.92\times10^{-1}$ & $9.82\times10^{-1}$ & $2.05\times10^{-2}$ \\
24 & $1.36\times10^{-3}$ & $9.24\times10^{-3}$ & $1.16\times10^{-1}$ & $2.21\times10^{-2}$ \\
32 & $1.87\times10^{-6}$ & $1.45\times10^{-5}$ & $1.20\times10^{-1}$ & $1.86\times10^{-2}$ \\
40 & $7.26\times10^{-9}$ & $6.95\times10^{-8}$ & $4.36\times10^{-2}$ & $2.87\times10^{-2}$ \\
fp & $1.79\times10^{-11}$ & n/a & n/a & n/a \\
\bottomrule
\end{tabular}
\end{table}

\section{Gate-level D2Q9 validation of the digital QLBM contract}
\label{app:d2q9}
D2Q9 is a second stencil of the same compute--reset architecture
(\cref{prop:cptp-reset}), with the moment basis of
\cite{lallemand2000theory}.
The discrete velocities and weights are
\begin{equation}
\begin{array}{c|ccccccccc}
i&0&1&2&3&4&5&6&7&8\\
\hline
c_{ix}&0&1&0&-1&0&1&-1&-1&1\\
c_{iy}&0&0&1&0&-1&1&1&-1&-1\\
w_i&\frac49&\frac19&\frac19&\frac19&\frac19&
\frac1{36}&\frac1{36}&\frac1{36}&\frac1{36}
\end{array}
\label{eq:d2q9-velocities}
\end{equation}
and
\begin{equation}
M=\begin{pmatrix}
1&1&1&1&1&1&1&1&1\\
-4&-1&-1&-1&-1&2&2&2&2\\
4&-2&-2&-2&-2&1&1&1&1\\
0&1&0&-1&0&1&-1&-1&1\\
0&-2&0&2&0&1&-1&-1&1\\
0&0&1&0&-1&1&1&-1&-1\\
0&0&-2&0&2&1&1&-1&-1\\
0&1&-1&1&-1&0&0&0&0\\
0&0&0&0&0&1&-1&1&-1
\end{pmatrix},
\label{eq:d2q9-M}
\end{equation}
with moment order $(\rho,e,\epsilon,j_x,q_x,j_y,q_y,p_{xx},p_{xy})$.
The collide--force--stream and bank-reset map is compiled to
executable logical Qiskit circuits from modular adder and multiplier
blocks, on a lattice small enough that the arithmetic remains simulable.
The three-dimensional HIT evidence is the D3Q19 register emulator of
Algorithm~\ref{alg:digital-timestep}.

Verification uses exact computational-basis simulation.
The basis simulator applies those adder and multiplier blocks by their integer semantics.
Every gate in $U_{\Phi}$ is a permutation of
computational-basis states, so propagating a single basis state reproduces
the circuit output exactly, at a cost of megabytes.
A tensor-network statevector of the same uncomputed workspace would need
tens of gigabytes.
That margin makes the last two checks in
\cref{tab:d2q9-validation} affordable.
The first is a two-timestep circuit, the shortest case that reuses a bank
after reset, and therefore tests the reset itself.
The second is the final state of every ancilla, the
direct witness that $U_{\Phi}$ is reversible.
In all cases the arithmetic workspace returns to $\ket{0}$.
The reset count equals one source population bank per timestep, and for LES also the obsolete shear-rate word that seeds the next Picard step.
For D2Q9 LES the digital word $x$ includes that persistent rate, so both resets are of the form assumed by \cref{prop:cptp-reset}.
Populations are encoded at the initial condition only.
A reversible digital $U_{\Phi}$ is evaluated into a clean destination bank
(MRT $\lambda$-multiply, body forcing when $A\neq0$, and the D2Q9 Chai/Zhang shear control of \cref{sec:d2q9-les} for LES).
The obsolete source bank is then unconditionally reset, and measurement occurs at the end.
\Cref{fig:qlbm-register-dns-les,fig:qlbm-register-lifetime} show that logical contract.
A hardware compilation of reset may use mid-circuit measurement and a corrective $X$.

\subsection{D2Q9 force-corrected LES strain}
\label{sec:d2q9-les}
Yu et al.\ reconstruct strain from nonequilibrium moments on D3Q19 \cite{yu2006mrtles}.
The gate-level D2Q9 LES uses the two-dimensional MRT stress reductions of Chai \cite{chai2012mrtstrain} and Zhang et al.\ \cite{zhang2018strain}, applied to force-corrected nonequilibrium moments.
Let $\mathbf{\Phi}_m=M\mathbf{\Phi}_f$ be the discrete force source in moment space, distinct from the Cartesian body force $(F_x,F_y)$.
The force-corrected nonequilibrium is
\begin{equation}
\widetilde{\delta\mathbf{m}}
=
\mathbf{m}-\mathbf{m}^{\mathrm{eq}}
+\frac{\Delta t}{2}\mathbf{\Phi}_m.
\end{equation}
The D2Q9 components are then
\begin{equation}
\begin{aligned}
S_{xx}
&=
-\frac{s_e\widetilde{\delta m}_e
+3s_\nu\widetilde{\delta m}_{p_{xx}}}{4\rho\Delta t},
\\
S_{yy}
&=
-\frac{s_e\widetilde{\delta m}_e
-3s_\nu\widetilde{\delta m}_{p_{xx}}}{4\rho\Delta t},
\\
S_{xy}
&=
-\frac{3s_\nu\widetilde{\delta m}_{p_{xy}}}{2\rho\Delta t},
\end{aligned}
\end{equation}
with $|S|=\sqrt{2(S_{xx}^2+S_{yy}^2+2S_{xy}^2)}$ and $\nu_t=(C_s\Delta)^2|S|$.
The unknown $s_\nu$ enters the strain formulas and is also $1/\tau_{\mathrm{eff}}$ after $\nu_t$ is formed, so the closure is implicit in $s_\nu$.
It is resolved by one Picard iteration per timestep, seeded by the previous shear rate and by $1/\tau_0$ at $t=0$.
On a Taylor--Green D2Q9 field after one DNS step the relative gap between that single Picard update and a twenty-iteration residual is $\mathcal{O}(10^{-4})$.
The classical D2Q9 solver, the integer twin, and the Qiskit circuit implement these same equations.
The D3Q19 production LES reconstructs strain with the molecular rate $1/\tau_0$ and does not use this Picard step.

\subsection{Arithmetic contract of the gate-level circuits}
\label{sec:d2q9-contract}
The D2Q9 datapath follows \cref{eq:fp-multiply} in its rounding except at
exact halfway cases, where the circuit rounds upward.
The signed multiply forms the exact product of the sign-extended operands,
adds $2^{f-1}$ to that product, and then copies the slice
$[f,f+w)$ into a clean output word.
The added constant turns the arithmetic
shift from a floor into a round, which removes the low bias of one unit in
the last place that a bare truncation would apply to every product.
Ties are resolved
upward here, which differs from the main-text
half-away-from-zero convention only on exactly representable halfway cases.
The rounding constant is added to the product workspace rather than to the
output word, so the multiply block remains its own inverse and can be used
directly as the uncompute of an earlier multiply.
Overflow uses two's-complement wraparound.
The production map of \cref{eq:fp-multiply} saturates.
In-circuit clamping would cost
a comparison and a controlled overwrite at every arithmetic stage.
The matched integer twin carries the same wraparound and is checked against the
declared word range, so a build that would have saturated is rejected.
At the widths in \cref{tab:d2q9-validation} every
stage stays inside the declared range, so the two conventions coincide on every
circuit listed.

\Cref{tab:d2q9-validation} reports the gate-level suite.
Every listed circuit is bit-exact against a matched fixed-point integer twin of the gate map.
DNS and LES with $A=0$ validate the D2Q9 collide/reset arithmetic and the
bank-reuse contract on $1\times1$ lattices, where periodic streaming is the
identity permutation
(LES includes the in-circuit D2Q9 Chai/Zhang strain of \cref{sec:d2q9-les}).
A separate DNS circuit with nonzero force amplitude validates forcing in-circuit
(\cref{fig:d2q9-forced-dns}).
These results confirm that the architecture is realizable as a
reversible gate-level $U_{\Phi}$ with unconditional reset of the obsolete
bank, on a stencil other than D3Q19.
The $320^3$/$64^3$ HIT comparisons are the long-horizon register
implementation of that map.

\begin{table}[H]
\centering
\caption{Gate-level D2Q9 validation of the same digital compute--reset contract
as \cref{eq:boxed-algo,prop:cptp-reset}, with encoding at the initial condition
only, a reversible digital $U_{\Phi}$, an unconditional reset of the
obsolete source bank, and measurement only at the end.
All circuit outcomes are bit-exact against the matched fixed-point integer twin.
The $T{=}2$ row is the first case that reuses a bank after resetting it.
The clean-workspace column records that every arithmetic ancilla returns to $\ket{0}$.
LES also keeps a persistent shear-rate word, reset with the obsolete source bank, which seeds the next Picard step.
D2Q9 is the gate-level check on a second stencil.
Long-horizon HIT remains the D3Q19 register-emulator evidence.}
\label{tab:d2q9-validation}
\small
\setlength{\tabcolsep}{3.5pt}
\begin{tabular}{@{}lccccccccc@{}}
\toprule
Case & $(n_x,n_y)$ & $T$ & $(w,f)$ & $A$ & Qubits & Resets & Bit-exact & Clean \\
\midrule
DNS ($F{=}0$) & $1{\times}1$ & 1 & $(7,4)$ & $0$ & 701 & 63 & yes & yes \\
LES ($F{=}0$) & $1{\times}1$ & 1 & $(7,4)$ & $0$ & 736 & 70 & yes & yes \\
DNS (forced) & $1{\times}1$ & 1 & $(8,5)$ & $0.25$ & 799 & 72 & yes & yes \\
DNS ($F{=}0$, $T{=}2$) & $1{\times}1$ & 2 & $(7,4)$ & $0$ & 701 & 126 & yes & yes \\
\bottomrule
\end{tabular}

\vspace{0.4em}
{\footnotesize
All cases use $\tau=0.8$ and a Taylor--Green initial condition with $u_0=0.05$,
verified on an exact computational-basis simulator, with measurement after
the final timestep.
All listed circuits are $1{\times}1$, so periodic streaming is the identity.
Forced DNS uses a uniform body force on the single cell, since the sinusoidal
sample vanishes at the $1{\times}1$ centre.
Logical depths and sizes are DNS ($F{=}0$) $6854$/$12846$, LES ($F{=}0$) $24899$/$43738$, DNS (forced) $19701$/$45360$, DNS ($F{=}0$, $T{=}2$) $13706$/$25622$.
The largest absolute population difference against classical float64 from the
same quantized initial condition is DNS ($F{=}0$) $0.030$, LES ($F{=}0$) $0.030$, DNS (forced) $0.022$, DNS ($F{=}0$, $T{=}2$) $0.148$.
The forced overlay in \cref{fig:d2q9-forced-dns} uses the same forced-DNS case.
}
\end{table}

\begin{figure}[H]
 \centering
 \includegraphics[width=0.96\linewidth]{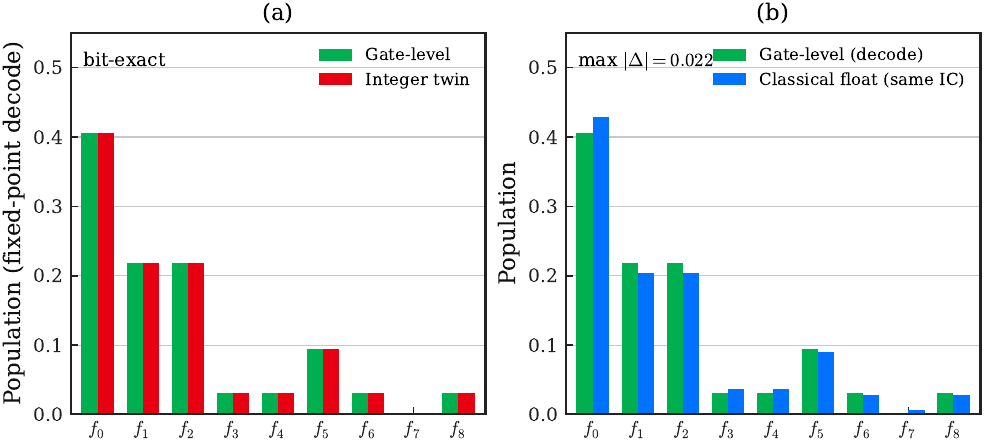}
 \caption{Forced D2Q9 gate-level DNS one-step check ($1{\times}1$, $T=1$, $(w,f)=(8,5)$, $A=0.25$).
 (a)~Decoded circuit populations coincide with the matched fixed-point integer twin (bit-exact).
 (b)~Same circuit decode versus classical float64 collide--force--stream from the identical quantized initial condition.
 The annotated $\max|\Delta|=0.022$ is that float64 gap.
 Streaming is trivial on this $1\times1$ lattice.}
 \label{fig:d2q9-forced-dns}
\end{figure}

\bibliographystyle{unsrtnat}
\bibliography{references}
\end{document}